\documentclass[onefignum,onetabnum]{siamart171218}

\usepackage{lipsum}
\usepackage{amsfonts}
\usepackage{graphicx}
\usepackage{epstopdf}
\usepackage{algorithmic}
\ifpdf
  \DeclareGraphicsExtensions{.eps,.pdf,.png,.jpg}
\else
  \DeclareGraphicsExtensions{.eps}
\fi

\usepackage{xcolor}
\usepackage{amsmath, amssymb}
\usepackage{mathtools}
\usepackage{commath}
\usepackage{subfig}
\usepackage[export]{adjustbox}
\usepackage{multicol}
\usepackage{vwcol}
\usepackage{tikz-cd}
\usepackage{listings}
\usepackage{hyperref}
\usepackage{bm}
\usepackage{cite}
\usepackage{enumitem}
\usepackage{rotating}
\makeatletter
\def\namedlabel#1#2{\begingroup
    #2%
    \def\@currentlabel{#2}%
    \phantomsection\label{#1}\endgroup
}
\makeatother

\newcommand{\scom}{simplicial complex}
\newcommand{\K}{\mathcal{K}}
\newcommand{\intr}{\textnormal{int}}

\newsiamremark{remark}{Remark}
\newsiamremark{hypothesis}{Hypothesis}
\crefname{hypothesis}{Hypothesis}{Hypotheses}
\newsiamthm{claim}{Claim}

\headers{Analysis of {Spatial and} Spatiotemporal Anomalies Using PH}{A. Hickok, D. Needell, and M. A. Porter}

\title{Analysis of {Spatial and} Spatiotemporal Anomalies Using Persistent Homology: Case Studies with COVID-19 Data
\thanks{Submitted to the editors on 26 July 2021.
\funding{AH and MAP acknowledge support from the National Science Foundation (grant number 1922952) through the Algorithms for Threat Detection (ATD) program. MAP also acknowledges support from the National Science Foundation (grant number DMS-2027438) through the RAPID program.}} {DN acknowledges support from the National Science Foundation (grant number DMS-2011140).}}

\author{Abigail Hickok\thanks{Department of Mathematics, University of California, Los Angeles, CA, USA
  (\email{ahickok@math.ucla.edu},
  \email{deanna@math.ucla.edu},
  \email{mason@math.ucla.edu}).}
\and Deanna Needell\footnotemark[2]
\and Mason A. Porter\footnotemark[2],\thanks{Santa Fe Institute, Santa Fe, NM, USA}
}

\usepackage{amsopn}

\makeatletter
\newcommand*{\addFileDependency}[1]{
  \typeout{(#1)}
  \@addtofilelist{#1}
  \IfFileExists{#1}{}{\typeout{No file #1.}}
}
\makeatother



\newcommand*{\myexternaldocument}[1]{%
    \externaldocument{#1}%
    \addFileDependency{#1.tex}%
    \addFileDependency{#1.aux}%
}

\ifpdf
\hypersetup{
  pdftitle={Analysis of Spatial and Spatiotemporal Anomalies Using Persistent Homology: Case Studies with COVID-19 Data},
  pdfauthor={A. Hickok, D. Needell, and M. A. Porter}
}
\fi

\myexternaldocument{supplement}

\begin{document}

\maketitle

\begin{abstract}
We develop a method for analyzing spatial and spatiotemporal anomalies in geospatial data using topological data analysis (TDA). To do this, we use persistent homology (PH), which allows one to algorithmically detect geometric voids in a data set and quantify the persistence of such voids. We construct an efficient filtered \scom\ (FSC) such that the voids in our FSC are in one-to-one correspondence with the anomalies. Our approach goes beyond simply identifying anomalies; it also encodes information about the relationships between anomalies. We use vineyards, which one can interpret as time-varying persistence diagrams (which are an approach for visualizing PH), to track how the locations of the anomalies change with time. We conduct two case studies using spatially heterogeneous COVID-19 data. First, we examine vaccination rates in New York City by zip code at a single point in time. Second, we study a year-long data set of COVID-19 case rates in neighborhoods of the city of Los Angeles.

\end{abstract}




\section{Introduction}

Many systems are spatial in nature. When working with spatial data sets, it is important to study the role of underlying spatial relationships \cite{geostats}. To illustrate this importance, consider the spatiotemporal dynamics of Coronavirus disease 2019 (COVID-19) case rates, which is one of the key motivations for our work. 
The spatial adjacencies between the neighborhoods of a city affect {the dynamics of disease spread \cite{porter2016},} and it is important to account for them. Researchers have studied a wide variety of spatial data sets, such as gross domestic product (GDP) and life expectancy by country \cite{gapminder, mindTheGap} and voting in elections across different regions of a state \cite{feng2021}. Such data sets often also include temporal information (e.g., {daily} COVID-19 case rates), and it is {also important to take it} into account.

We develop new methods for using \emph{topological data analysis} (TDA) to analyze {geospatial and geospatiotemporal} data sets in a way that directly incorporates spatial information. TDA is a way {to study} the ``shape'' of a data set \cite{gunnar-natphys}. Using \emph{persistent homology} (PH), {which is} {an approach} from algebraic topology, {one can algorithmically} find {geometric} voids of different dimensions in a data set {and} quantify the ``persistence'' of these voids \cite{otter2017}. Zero-dimensional (0D) voids are connected components {and} one-dimensional (1D) voids are holes. To quantify the persistence of holes and other voids, one constructs a \emph{simplicial complex} ({which is} a combinatorial description of a topological space) and a \emph{filtration function} (see section \ref{sec:PH}). {In our work, we treat geographical data} as {two-dimensional (2D)} data and construct a 2D filtered \scom\ {(FSC)} to represent it. The computation of PH has yielded insights into a wide variety of areas, such as dynamical systems \cite{chaos_dynamics,chaos_symbolic}, collective behavior \cite{topaz}, neuroscience \cite{neuro_sizemore, neuro_giusti}, materials science \cite{materials}, and chemistry \cite{cyclo}. Spatial applications that have been examined as 2D data sets {using PH} include sensor networks \cite{sensor}, percolation, \cite{disk_percolation}, and city-street networks and other complex systems \cite{feng2020}.

{When we examine} time-dependent data, we use \emph{vineyards}, which were introduced in \cite{vineyards} as a way {to represent time-varying} PH, to incorporate temporal information. One can visualize a vineyard as a continuous stack of persistence diagrams (PDs), with one PD for each time point. The homology classes trace out curves, which are called \emph{vines}, in $\mathbb{R}^3$. At any single point {in time}, a homology class in the PD at that {time} corresponds to a (birth simplex, death simplex) {pair. The} birth simplex creates the homology class, and the death simplex destroys the homology class (see section \ref{sec:PH}). In a vineyard, a vine corresponds to a sequence of (birth simplex, death simplex) pairs. See section~\ref{sec:vineyard} for the definition of a vineyard.


\subsection{Our Contributions}\label{sec:contributions}

We use TDA to analyze local extrema of real-valued geospatial data\footnote{See {s}ection \ref{sec:filtration} for our definition of a ``local maximum'' and a ``local minimum'' of a real-valued function on a discrete set of geographical regions.}. Our approach captures both local information (specifically, the geographical locations and the values of the local extrema) and global information about the relationships between the {extrema. The global information includes the} extent to which extrema are ``spatially separated'' (see Figure \ref{fig:separation}).

\begin{figure}
    \centering
    \subfloat[]{\includegraphics[width = .45\textwidth]{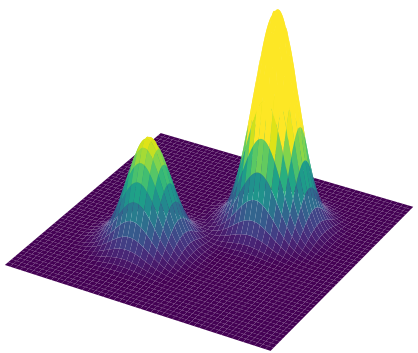}\label{subfig:separation}}
    \hspace{5mm}
    \subfloat[]{\includegraphics[width = .45\textwidth]{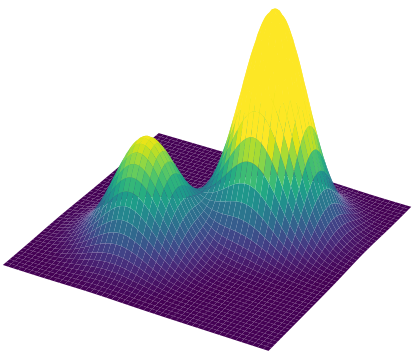}\label{subfig:nonseparation}}
    \caption{(a) The graph of a function $f: \mathbb{R}^2 \to \mathbb{R}$ that has two ``well-separated'' local maxima. (b) The graph of a function $g: \mathbb{R}^2 \to \mathbb{R}$ whose two local maxima have the same locations and values as $f$, {but are not well-separated from} each other.}
    \label{fig:separation}
\end{figure}

To the best of our knowledge, {existing} methods of analyzing {local} extrema yield only local information. {One can check} {whether} a geographical region is an extremum by {comparing} its value to those of its {neighboring regions}. However, this approach does not provide any global information about the extrema. {For example,} it cannot distinguish between the two cases {in }Figure \ref{fig:separation}.

{Examining} vineyards allow us to measure the persistence {of} extrema {with} time, observe {how} spatial {separations} between extrema {change} {with} time, and track {how} geographical locations {of} extrema change {with} time. We accomplish the {last of these} by using the vineyards to match the extrema at one time to the corresponding extrema at another {time. (They may not be at the same geographical locations.)} We identify the geographical locations {of} extrema by {examining} the sequence of (birth simplex, death simplex) pairs for each vine. To the best of our knowledge, {the present paper} is the first {paper} that uses information about the sequence of (birth simplex, death simplex) pairs for each vine, rather than using {only} the (birth, death) filtration values for each vine. {A naive approach, such as comparing each region to its {neighboring regions} at each time step, does not come with a natural way {to match} the extrema {that one identifies} at different {times} and does not provide information about changes in global structure. {With} our approach,} we are able to track how the global spatial structure {of} data changes {with} time.

Another contribution of our paper is a {new} method to construct an ``efficient'' simplicial complex whose underlying space\footnote{{The \emph{underlying space} of a \scom\ is the union of its simplices. We note that it is common in studies of TDA for authors to conflate the combinatorial and topological structures of a simplicial complex.}} is homeomorphic to a geographical space (which is the set of regions, as we will explain shortly)\footnote{The \scom\ is ``efficient'' in the sense that it minimizes the number of simplices.}.
{In {our applications, we possess geographical data in the form of {\sc shapefile}s}. Each geographical region (e.g., a neighborhood or zip code) is represented in {a} {\sc shapefile} by a polygon (or by multiple polygons, if the region is disconnected) with {many} vertices (about 100--1000 vertices, depending on the particular {\sc shapefile} and the particular region). {These} polygons approximate the real-life boundaries of the geographical regions. A naive approach to building a simplicial complex is to simply triangulate each of the polygons. However, this approach has two issues. The first is that there are often small overlaps between the polygons or spurious gaps between the polygons because the polygon boundaries {do not exactly match} the real-life geographical boundaries. The vertices of a polygon often lie in the interior of another polygon. The second issue is that simply triangulating {these} polygons, which each have a very large number of vertices, would create orders-of-magnitude more simplices than are necessary to represent a geographical space.} It is important to attempt to minimize the number of simplices in a simplicial complex because PH and vineyard computation times are very sensitive to the number of simplices.

{Rather than naively triangulate the given polygons, we use the {\sc shapefile} of a geographical space to infer adjacency information about the regions; we then use only this information to build a simplicial complex for that geographical space. In {the resulting} simplicial complex, each region is represented by a union of triangles. We use about 1--10 triangles per region, depending on the number of neighbors {of the region}{. By} contrast, the naive approach above {requires} about 100--1000 triangles per region. Two adjacent regions that have a connected intersection share exactly one edge in our simplicial complex, except in rare special cases that we will discuss in section \ref{scom_construction}.}
In our simplicial complex for a geographical space, the union of any subset of geographical regions is homeomorphic to the {underlying space of the} {simplicial subcomplex} (see section \ref{scom_construction} for the definition of a {simplicial subcomplex}) that is induced by the union of the corresponding {triangles}. {When the geographical regions satisfy the mild assumptions \ref{geo_firstcond}--\ref{geo_lastcond} that we define in section \ref{scom_construction},} our construction uses the {minimum} number of simplices that {is possible for a \scom\ with the property above. {(See Property \ref{scom_cond} in section \ref{scom_construction}}.)}

As case studies, we apply our {approach} to two data sets. The first data set is a geospatial data set of per capita vaccination rates in New York City (NYC) by {zip code \cite{NYCvax}}. The homology classes {correspond to} zip codes in which the vaccination rate is either lower or higher (depending on choices that one can make in our approach) than in the {neighboring} zip codes. {The estimates of these rates are at a single point in time (23 February 2021).} The second data set consists of 14-day mean per capita COVID-19 case rates in neighborhoods in the city of Los Angeles (LA) in the time period 25 April 2020--25 April 2021. Modeling the spatiotemporal spread of COVID-19 is a complex task \cite{arino2021describing, vesp2020}. In this geospatiotemporal data set, the homology classes of our approach correspond to COVID-19 anomalies, which are regions whose case rates are higher than in the {neighboring} regions.\footnote{We examine \emph{local} maxima in the case-rate data. This contrasts to COVID-19 ``hotspots,'' which the CDC has defined using an absolute threshold for the number of cases and criteria that are related to the temporal increase in the number of cases \cite{cdc_hotspot}.} It is important to examine such anomalies, as COVID-19 spreads with significant spatial heterogeneity and thus has heterogeneous effects on different areas.\footnote{Other scholars have studied contagions using TDA in ways that do not yield topological features with geographical meaning. For example, recent work used TDA to study the spatiotemporal spread of COVID-19 \cite{yuliagel_covid} and Zika \cite{zika}. These papers {examined} topological features in atmospheric data, which were then used to forecast case rates. TDA was also used in \cite{taylor} to {study} the Watts threshold model of a social contagion on noisy geometric networks.
}
Many factors (such as mobility, population density, socioeconomic differences, and racial demographics) play a role in how COVID-19 affects regions differently \cite{mobilitypopdens, spatial_hetero, race_covid}. In our case study of COVID-19 case rates in LA, we construct a vineyard that (1) {conveys} which anomalies are most persistent in time and (2) reveals how the anomalies move geographically with time.


\subsection{Related Work}

Our method addresses several limitations of previous efforts to combine TDA with geospatial analysis. In \cite{brexit}, Stolz {et al.} studied the percentage of United Kingdom voters by district that voted to leave the European Union in {the} ``Brexit'' referendum. The holes that they identified using PH corresponded to districts that voted differently than the {neighboring} districts. However, {their approach} {does} not distinguish between homology classes that were merely noise and homology classes that corresponded to small geographical districts. In \cite{feng2021}, Feng and Porter developed an {approach to study PH by constructing FSCs} using the level-set method {\cite{osher2003}} of front propagation {from scientific computation}\footnote{The name ``level-set method'' may cause confusion. Importantly, the level-set simplicial complex of \cite{feng2021} is not the simplicial subcomplex that {has} simplices with some prescribed filtration value (i.e., a level set of the filtration values of {a} simplicial complex).}{.} Using their level-set complexes, they examined the percentage of voters in each precinct of California counties that voted for a given candidate (e.g., Hillary Clinton) in the 2016 {United States} presidential election. The homology classes represent precincts that voted more heavily for Clinton than the {neighboring precincts;} {these regions} are ``islands of blue in a sea of red''. The level-set complexes in \cite{feng2021} have two key limitations. The first is that they cannot handle time-dependent data, as they are built to study {either} data at a single point in time or data that has been aggregated over some time window to yield time-independent data. The second limitation is that {these \scom es} reduce real-valued data (e.g., the percentage of voters who voted for Clinton) to binary data (e.g., whether or not the majority voted for Clinton). Consequently, in this example, the level-set-based PH does not capture the extent to which a blue ``political island'' voted more heavily for Clinton. By contrast, our approach is designed specifically to capture such information. As a trade-off, we no longer capture the geographical sizes of the political islands. {For further discussion, see} Feng, Hickok, and Porter \cite{tda_spatial}, who applied the level-set filtration to study 
{the cumulative case count in Los Angeles on one specific day.}

Our new approach {to compute} PH is also able to resolve some other technical issues in \cite{feng2021}. In particular, some of the homology classes in the level-set approach {of} \cite{feng2021} are geographical artifacts that are indistinguishable from true features of a data set. {By contrast, the finite 1D homology classes in our approach are either in} one-to-one correspondence with the local maxima of {a} real-valued geospatial function or {in} one-to-one correspondence with {its} local minima, depending on the choices that one makes. Additionally, unlike the level-set approach in \cite{feng2021}, we are able to detect extrema that are adjacent to the boundary of {a} geographical space.

Other methods {to construct} simplicial complexes from geospatial data, such as rasterization of a {\sc shapefile} {or} treating the regions as a point cloud, require a trade-off between the number of simplices and the accuracy of the representation of the geographical regions. For example, the level-set-based PH method of \cite{feng2021} uses orders-of-magnitude more simplices to achieve sufficient resolution of the smallest geographical regions (e.g., densely populated urban centers that are important to analyze). See section~\ref{sec:conclusion} for further discussion.

We use vineyards in the present paper, but there are also other ways to study the topology of time-varying data. For example, zigzag PH \cite{zigzag} was used in \cite{CorcoranJones} to analyze time-dependent point clouds (such as swarms) and in \cite{munch} to study time-delay embeddings of dynamical systems. Crocker plots and crocker stacks (i.e., stacks of crocker plots for different values of a smoothing parameter) illustrate how the Betti numbers of a time-dependent point cloud change with time and with a scale parameter $\epsilon$ \cite{crocker}. {Additionally,} Kim and M\'emoli \cite{kim_memoli} used multiparameter PH \cite{multi} to study time-dependent point clouds. In appendices \ref{sec:multi} and \ref{sec:zigzag}, we show how one can use multiparameter PH \cite{multi} and multiparameter zigzag PH \cite{zigzag} to {study} our COVID-19 spatiotemporal data sets.


\subsection{Organization}

Our paper proceeds as follows. In section \ref{sec:background}, we {briefly} review relevant topological background. In section~\ref{scom_construction}, we formulate how we construct simplicial complexes. In section~\ref{sec:filtration}, we define several filtration functions and discuss how to interpret the resulting PDs and vineyards. In section~\ref{sec:apps}, we apply our method to the LA and NYC data sets. In section~\ref{sec:discussion}, we discuss our {methodological choices}. In section \ref{sec:conclusion}, we summarize our work and discuss some of its implications. In the appendix, we discuss technical details of the \scom\ {construction}, discuss alternative topological approaches for studying PH in geospatiotemporal data, {provide further information about the LA results, compare our approach to an ``all-but-one'' statistical test, and show some demographic data.} Our code is available at \url{https://bitbucket.org/ahickok/vineyard/src/main/}.


\section{Background}\label{sec:background}


\subsection{Persistent Homology (PH)}\label{sec:PH} 

We {give a brief introduction to} persistent homology (PH). See \cite{otter2017} for a more thorough {discussion of it}.

{A \emph{$k$-simplex} is $k$-dimensional polytope that is a convex hull of $k+1$ vertices. A convex hull of a subset of these vertices is a \emph{face} of the simplex.} A \emph{simplicial complex} $\K$ is a set of simplices that satisfies two requirements: (1) if $\sigma \in \K$ is a $k$-simplex, then every face of $\sigma$ is in $\K$; (2) if $\sigma$ and $\tau$ are simplices in $K$, then $\sigma \cap \tau$ is a face of both $\sigma$ and $\tau$.

{A \emph{filtered simplicial complex} (FSC) is a nested sequence $\K_{\alpha_0} \subseteq \cdots \subseteq \K_{\alpha_n} = \K$ of \scom es for {some sequence $\{\alpha_0, \ldots, \alpha_n\}$ of indices}. See Figure~\ref{fig:fsc} for an example of an FSC.} A \emph{filtration function} (or simply a \emph{filtration}) $f: \K \to \mathbb{R}$ is a function such that if the simplex $\tau \in \K$ is a face of $\sigma \in \K$, then $f(\tau) \leq f(\sigma)$. {A pair $(\K, f)$ induces an FSC as follows.} Let $\K_{\alpha} := \{ \sigma \in \K \mid f(\sigma) \leq \alpha \}$ be the \emph{$\alpha$-sublevel \scom}, and let $\{\alpha_0, \ldots, \alpha_n\}$ be the image of $f$. The sequence $\K_{\alpha_0} \subseteq \cdots \subseteq \K_{\alpha_n} = {\K}$ is a nested sequence of simplicial complexes. {In our paper, we often refer to the pair $(\K, f)$ as the FSC itself. We do this because it is the most natural way to define the FSCs for our applications.}

\begin{figure}
\begin{minipage}{.5\textwidth}
    \centering
    \subfloat[$\K_0$]{\includegraphics[width = .2\linewidth]{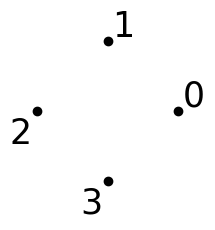}}
    \subfloat[$\K_1$]{\includegraphics[width = .2\linewidth]{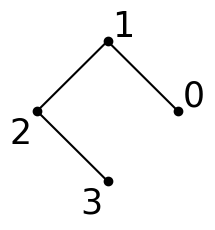}}
    \subfloat[$\K_2$]{\includegraphics[width = .2\linewidth]{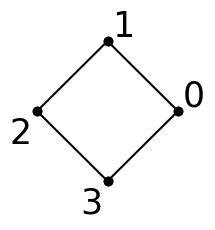}}
    \subfloat[$\K_3$]{\includegraphics[width = .2\linewidth]{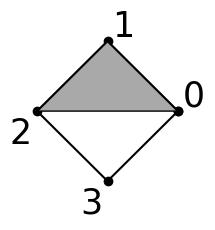}}
    \subfloat[$\K_4$]{\includegraphics[width = .2\linewidth]{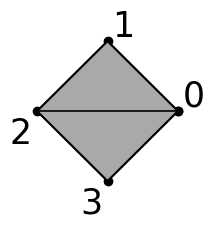}}
    \caption{An example of nested simplicial complexes in a filtered simplicial complex.
    }

    \label{fig:fsc}
\end{minipage}%
\hspace{5mm}
\begin{minipage}{.45\textwidth}
    \centering
    \includegraphics[width = \linewidth]{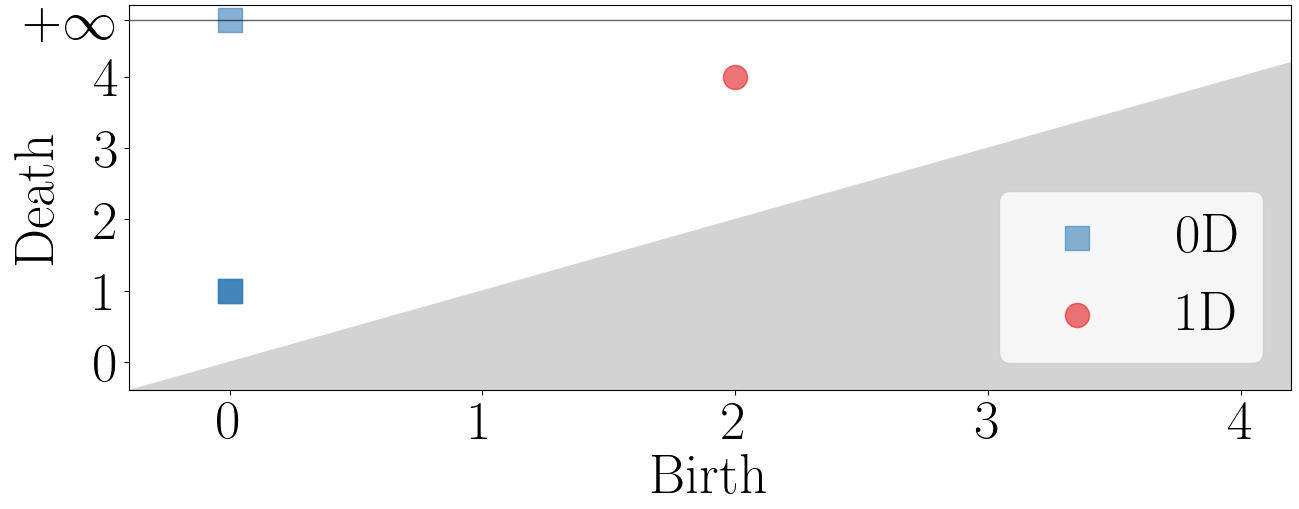}
    \caption{The persistence diagram {of} the filtered simplicial complex in Figure~\ref{fig:fsc}.}
    \label{fig:fscPD}
\end{minipage}
\end{figure}

We compute the homology of each $K_{\alpha_i}$ over a field $\mathbb{F}$, which we set to $\mathbb{Z}/2\mathbb{Z}$ in the present paper. {Let $H_p(\K_{\alpha_i}, \mathbb{F})$ denote the $p$-dimensional homology of $\K_{\alpha_i}$ over $\mathbb{F}$.} Homology classes represent connected components, holes, and higher-dimensional voids in a simplicial complex; {specifically,} $p$-dimensional homology classes represent $p$-dimensional {``holes.''} The inclusion {relationship $\K_{\alpha_i} \xhookrightarrow{} \K_{\alpha_{i+1}}$ between subcomplexes} induces a map $\iota_i: {H_p}(\K_{\alpha_i}, \mathbb{F}) \to {H_p}(\K_{\alpha_{i+1}}, \mathbb{F})$ from the {homology of $\K_{\alpha_i}$ to the homology of $\K_{\alpha_{i+1}}$.} The inclusion map $\iota_i$ lets us track an element of {$H_p(\K_{\alpha_i}, \mathbb{F})$} to an element of {$H_p(\K_{\alpha_{i+1}}, \mathbb{F})$}. {The $p$-dimensional {PH} of an FSC is the pair}
\begin{equation}
    \left(\{{H_p}(\K_{\alpha_i}, \mathbb{F})\}, \{\iota_i\} \right)_{0\leq i < n}\,.
\end{equation}
We say that a {homology class $\gamma$ is \emph{born}} at filtration level $\alpha_i$ if {$i$ is the {smallest index for} which $\gamma$ is an element of $H_p(\K_{\alpha_i}, \mathbb{F})$.} We say that the homology class $\gamma$ \emph{dies} at filtration level $\alpha_j$ if $\alpha_{j-1}$ is the last filtration level at which $\gamma$ exists. That is, $\iota_{j-1} \circ \cdots \circ \iota_i$ maps $\gamma \in H_p(\K_{\alpha_i}, \mathbb{F}$) to $0$ in $H_p(\K_{\alpha_j}, \mathbb{F})$ and for all $k < j-1$, we have $\iota_k \circ \cdots \circ \iota_i(\gamma) \neq 0$. Not every homology class dies; we refer to classes that do die as \emph{finite} and classes that do not die as \emph{infinite}.

The Fundamental Theorem of Persistent Homology yields a set of generators for the {PH} of an FSC \cite{edel_book, dey_book}. Each generator is a homology class. A generator has a \emph{birth simplex} $\sigma_b$ that creates the homology class and (if finite) a \emph{death simplex} $\sigma_d$ that destroys the homology class. If one is computing homology in dimension $p$, then $\sigma_b$ is a $p$-dimensional simplex and $\sigma_d$ is a $(p+1)$-dimensional simplex. The simplex pair $(\sigma_b, \sigma_d)$ represents the homology class. For example, in Figure~\ref{fig:fsc}, there is one 1D PH generator. Its birth simplex is the edge $(0, 3)$ because this is the edge that completes the loop {that encircles the hole}, and its death simplex is the triangle $(0, 2, 3)$ because this is the triangle that fills in the {hole}. The birth filtration level of the homology class is $f(\sigma_b)$ {and} the death filtration level (if finite) is $f(\sigma_d)$.

A \emph{persistence diagram} (PD) is a way of representing PH as a multiset of points in the {extended plane $\overline{\mathbb{R}}^2$}. Each {off-diagonal} point represents a {generator of the {PH}; the point's coordinates are the {homology} class's birth and death filtration levels. One includes the points on the diagonal for technical reasons; one can think of them as homology classes that die instantaneously upon birth. See Figure~\ref{fig:fscPD} for an example of a PD.


\subsection{Vineyards}\label{sec:vineyard}

{The examination of \emph{vineyards} is one way to study} time-varying PH \cite{vineyard}. A \emph{time-dependent filtration function} on a simplicial complex $\K$ is a function $f:[t_0, T] \times \K \to \mathbb{R}$ {such that $f(t, \cdot)$ is a filtration for all $t$.} We compute the PH of $(\K, f(t, \cdot))$ for all times $t$. We visualize the vineyard in $\mathbb{R}^2 \times [t_0, T]$ as a continuous stack of PDs (see Figure~\ref{fig:vineyard}). The points in the PDs trace out curves {with} time; these curves are the \emph{vines}. Each vine corresponds to a homology class; a vine is the graph of the birth and death filtration levels of a particular homology class {as a function of} time. The homology class that is represented by a vine has a \emph{time-dependent birth simplex} $\sigma_b(t)$ and (if finite) a \emph{time-dependent death simplex} $\sigma_d(t)$. At time $t$, the homology class is created by the simplex $\sigma_b(t)$ at filtration level $f(t, \sigma_b(t))$ and {(if finite) destroyed by the simplex $\sigma_d(t)$ at filtration level $f(t, \sigma_d(t))$}. The functions $\sigma_b(t)$ and $\sigma_d(t)$ are piecewise constant. We measure the overall \emph{persistence} of a vine by calculating $\int_{t_0}^T [f(t, \sigma_d(t)) - f(t, \sigma_b(t))]\, \mathrm{d}t$.

Cohen-Steiner et al. \cite{vineyards} developed an algorithm for computing vineyards when they {introduced them.} One computes the initial PH at time $t = t_0$, and one then updates {the pairings of birth and death simplices} as the order of the simplices (as induced by $f(t, \cdot)$) changes {with} time. Each change in the order of the simplices occurs one transposition at a time. One can make these updates} in $O(m)$ time (where {$m$} is the number of simplices) per transposition of simplices.

\begin{figure}
    \centering
    \includegraphics[width= .4\textwidth]{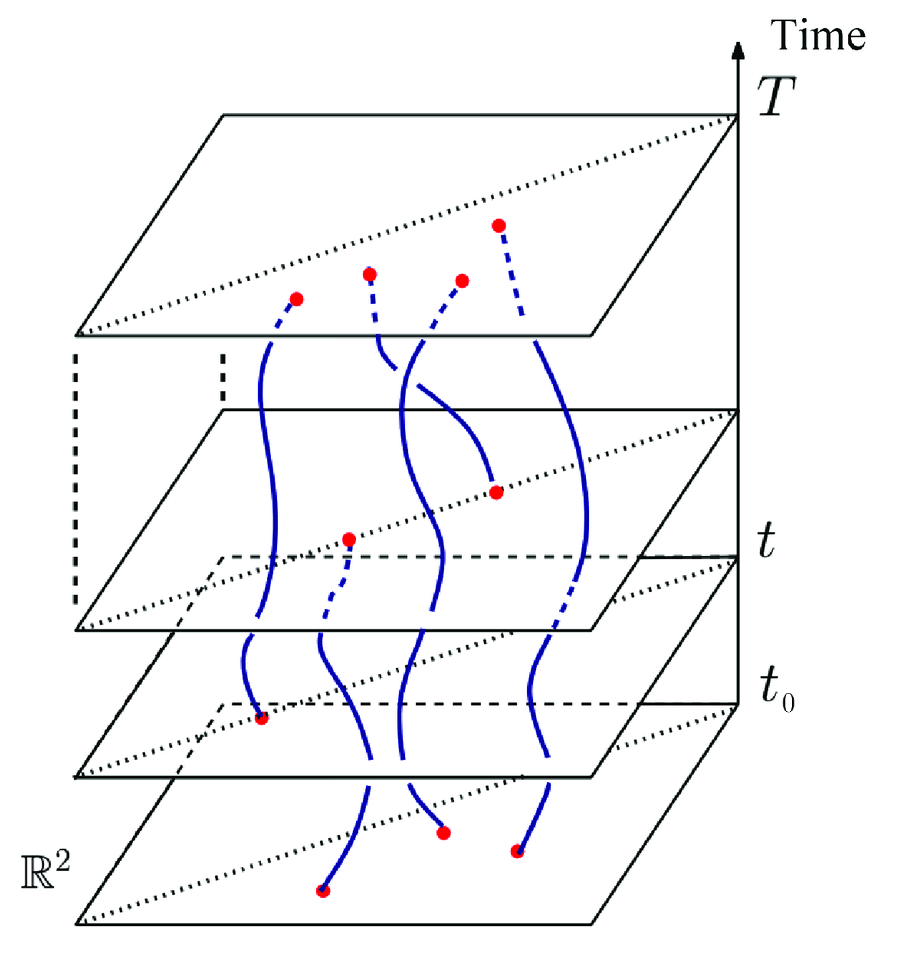}
    \caption{An example of a vineyard. Each curve is a vine in the vineyard. This figure is a slightly modified version
    of a figure that appeared originally in \cite{vineyard}.
    }
    \label{fig:vineyard}
\end{figure}


\section{Constructing a Simplicial Complex}\label{scom_construction}

We now show how we construct a simplicial complex $\K$ from geographical data (e.g., a {\sc shapefile} that specifies approximate geographical boundaries {of a set of geographical regions}). We partition {a given} geographical space into \emph{regions}. In section~\ref{sec:nyc}, the regions are zip codes in NYC; in section~\ref{sec:LA}, the regions are neighborhoods in the city of LA. Let $S$ be the set of regions. We refer to the complement of $\bigcup_{R \in S} R$ as the \emph{exterior region}. We construct a 2D simplicial complex $\K$ with the following property:
\begin{enumerate}
    \item[\namedlabel{scom_cond}{(P)}] {There is an assignment of 2D simplices to regions such that the union of any subset of regions is homeomorphic to the {underlying space of the} \emph{simplicial subcomplex}\footnote{The {simplicial subcomplex} that is induced by {a set $E \subseteq \K$} is 
    the smallest simplicial complex $\K'$ that contains the set $E$ of simplices. That is, if $\K''$ is a simplicial complex that contains $S$, then $\K' \subseteq \K''$. When $\K$ is 1D, a {simplicial subcomplex} is equivalent to an induced subgraph.} that is induced by the union of the corresponding 2D simplices.}
\end{enumerate}
In Figure~\ref{fig:exampleK}, we {show} an example of our construction{, which we discuss in this section and {present} in more detail in Appendix \ref{sec:details}. Under the mild assumptions \ref{geo_firstcond}--\ref{geo_lastcond} that we define shortly, our \scom\ has the minimum number of simplices that is possible for a \scom\ {that satisfies property \ref{scom_cond}.}}
Constructing an efficient \scom\ is important because {the time that it takes for} TDA computations {depends sensitively on} the number of simplices {in a \scom}.

\begin{figure}
	    \centering
	    \includegraphics[scale=.45]{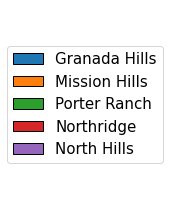}
	   \subfloat[]{\includegraphics[scale=.4]{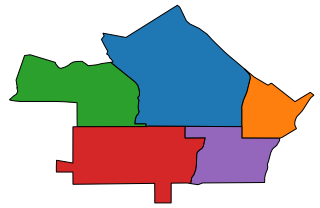}\label{fig:geo_granada}}
	    \hspace{5mm}
	    \subfloat[]{\includegraphics[scale=.35]{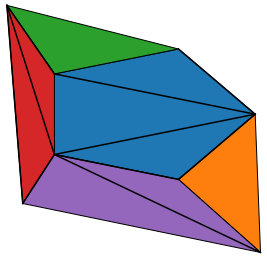}\label{fig:K_granada}}
	    \caption{(a) A set $S$ of geographical regions{, as given by a \sc{shapefile} \cite{lashp}.} (b) The resulting \scom\ $\K$.}
	    \label{fig:exampleK}
	\end{figure}

In our case studies, the geographical data take the form of {\sc shapefile}s. In a {\sc shapefile}, each region is represented by a {\emph{polygon with holes}\footnote{ A polygon with holes {is} $P = Q - \bigcup_{i=1}^h \intr(H_i)$, where $X$ is a polygon that encloses {polygons} $H_1, \ldots, H_h$ (the holes) \cite{artgallery} {and} $\intr(H_i)$ denotes the interior of $H_i$. {It is possible to have $h=0$ holes.}} (or by multiple polygons with holes}, if the region is disconnected) that closely approximates the actual geographical region. {(A {\sc shapefile} stores the coordinates of the boundaries of the polygons.)} For an example of {\sc shapefile} data, see Figure \ref{fig:geo_granada}. As we discussed in section \ref{sec:contributions}, the polygon boundaries are not always aligned perfectly, so their interiors sometimes overlap {and gaps can occur between them}. Therefore, to construct a \scom\ $\K$, we must do more than merely triangulate these polygons. Additionally, the polygons in our {\sc shapefile}s have roughly between 100 and 1000 vertices, which is many more vertices per region than in {the \scom\ $\K$ that we {will} construct shortly.}

We make the following assumptions about geographical regions:
\begin{enumerate}
    \item[\namedlabel{geo_firstcond}{(A1)}] There are a finite number of regions, and each region has a finite number of connected components.
    
    \item[\namedlabel{geo_bdrycond}{(A2)}] {Each {component of a region} is homeomorphic to $D_0 - \bigcup_{i=1}^h \intr(D_i)$, where $D_0$ is a closed disk that encloses some number (which can be $0$) of other closed disks $D_1, \ldots, D_h$ ({i.e.,} the holes of the region). {For all $i \neq j$, the intersection $D_i \cap D_j$ {has} at most one point.}
    See, for example, the West Vernon region in Figure \ref{fig:geo_westvernon}{; it is} homeomorphic to $D_0 - D_1$ {(an annulus)} for two disks $D_0$ and $D_1$ that do not intersect. {(In our case studies, it is rare for any of the disks to intersect.)}}
    
    \item[(A3)] The intersection between any two regions has a finite number of components, and the interiors of the regions do not intersect.
    
    \item[\namedlabel{geo_lastcond}{(A4)}] The intersection between three or more regions is either a point or $\emptyset$.
\end{enumerate}

Assumptions \ref{geo_firstcond}--\ref{geo_lastcond} are very reasonable for human-made geographical boundaries. We do not even require the regions to be simply connected or the region intersections to be connected. In Figure~\ref{fig:geo_granada}, we illustrate the most typical situation that we encounter. In this example, LA neighborhood Granada Hills is homeomorphic to a disk and its boundary intersects the boundaries of five neighboring regions (counting the exterior region). In Figures~\ref{fig:georegions} and~\ref{fig:geo_koreatown}, {we illustrate a few other configurations that can arise in geospatial applications.} 

\begin{figure}
    \centering
    \subfloat[Valley Glen]{\includegraphics[width = .45\textwidth]{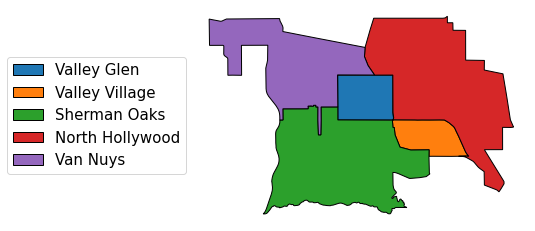}\label{fig:geo_valleyglen}}
    \hspace{10mm}
    \subfloat[West Vernon]{\includegraphics[width = .45\textwidth]{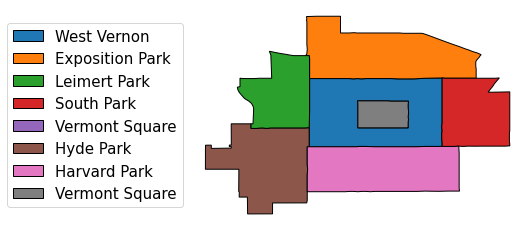}\label{fig:geo_westvernon}}
    \caption{Various neighborhoods {of Los Angeles}{, as given by a {\sc shapefile} \cite{lashp}.}
    (a) {The four neighborhoods Valley Glen, Valley Village, Sherman Oaks, and North Hollywood intersect in a point.} (b) The neighborhood West Vernon has a hole because of its neighbor Vermont Square.
    }
    \label{fig:georegions}
\end{figure}

We now outline our procedure for building a \scom.
{For each region $R$, we construct a {``reduced''} polygon with holes $P^R$ that has orders-of-magnitude fewer vertices than the polygons with holes in the associated {\sc shapefile}. The number of holes in $P^R$ is equal to the number of holes of the geographical region $R$. We {glue} the boundaries of $\{P^R \mid R \in S\}$ together in a way that respects the geographical region boundaries. {We then} triangulate each of the polygons to obtain a 2D \scom\ $\K$.} We assign a 2D simplex $\sigma \in \K$ to the region $R$ whose polygon $P^R$ originally contained $\sigma$. In Figure~\ref{fig:exampleK}, we show an example of {the resulting $\K$.} Our code for our simplicial-complex algorithm is available at \url{https://bitbucket.org/ahickok/vineyard/src/main/}.\footnote{This code has one limitation that the algorithm in the present paper does not{.} It requires that no interior region (i.e., a region that is contained {in} the outer boundary of another region) intersects any other interior region. This does not occur in our data, and we believe that it does not occur in most geographical spaces.} In the remainder of this section, we discuss the details of this process.


\subsection{Constructing a Reduced Polygon with Holes for each Region}

Without loss of generality, we assume that each region is connected; if not, we treat each component of a region as if it were its {own} region. For each region $R$, we construct a reduced {polygon with holes} $P^R$ using only adjacency information that we infer from a {\sc shapefile}. {Let $D_0, D_1, \ldots, D_h$ be the disks in the statement of assumption \ref{geo_bdrycond}, and let $B_i = \partial D_i$.} Under the geographical assumptions~\ref{geo_firstcond}--\ref{geo_lastcond}, the intersections of a region $R$ with its neighbors are such that for each {boundary $B_i$}, one can order the neighbors in clockwise (or counterclockwise) fashion, possibly with repetition\footnote{Theoretically, several 0D intersections can be adjacent to each other, although this scenario does not occur in our data sets. That is, in principle, there can {exist} a sequence $\{N_i, \ldots, N_{i+k}\}$ of neighbors such that $N_j \cap R $ is the same point $p$ for all $j$. The order of this sequence is not determined uniquely by the intersections of the neighbors with $R$. Instead, we order them in the order in which they appear clockwise (or counterclockwise) around the point $p$. This sequence must be finite because there are a finite number of regions and \ref{geo_bdrycond} implies that $N_{j_1} \neq N_{j_2}$ if $j_1 \neq j_2$.}. {Let $S_i$ denote this sequence of neighbors around $B_i$.} We list intersections with the exterior region in the same manner as for any other neighboring region. We also record whether each intersection is 1D or 0D. For example, in Figure \ref{fig:geo_valleyglen}, the clockwise sequence of neighbors around the boundary of Valley Glen is \{Van Nuys, North Hollywood, Valley Village, Sherman Oaks\}. The intersection with Valley Village is 0D and the other intersections are 1D. For regions such as West Vernon in Figure \ref{fig:geo_westvernon}, we obtain a sequence {$S_i$} for each boundary {$B_i$}. Each sequence is unique up to the choice of starting neighbor.

Given a sequence of neighbors for each {boundary $B_i$} (which, if necessary, we adjust as in Appendix~\ref{appendix:nbr_adjustment}), we construct a {polygon with holes $P^R$} as follows. {Let $(P')^R$ be a polygon that has one edge for each $N \in S_0$ such that the corresponding component of $N \cap B_0$ is 1D. Let $\{H_i^R\}_{i=1}^{h}$ be a {set} of polygons that are contained in $(P')^R$ and satisfy the following properties:
{
\begin{enumerate}
    \item $H_i^R$ has one edge 
    {for each $N \in S_i$ such that the corresponding component of $N \cap B_i$ is 1D},
    
    \item $H_i^R \cap H_j^R \neq \emptyset$ if and only if $D_i \cap D_j \neq \emptyset$, 
    \item $P^R \cap H_i^R \neq \emptyset$ if and only if $D_0 \cap D_i \neq \emptyset$, and
    \item {if} the intersection of two polygons in $\{P^R, H_1^R, \ldots, H_{h}^R\}$ is nonempty, then the intersection is a vertex.
\end{enumerate}
}
The locations of the vertices do not matter. We define $P^R$ to be $(P')^R - \bigcup_{i=1}^{h}\intr(H_i^R)$, which is homeomorphic to $R$ by assumption \ref{geo_bdrycond}. Finally, we annotate each edge of $P^R$ with the neighbor that corresponds to it. We also annotate each vertex with the sequence of its adjacent regions, which we list in clockwise order starting with $R$.}


\subsection{Gluing Together the Polygons with Holes}

We glue the {{polygons with holes} $\{P^R \mid R \in S\}$} along their edges according to their edge and vertex annotations. More precisely, if $P^{R_1}$ has $n$ nonadjacent edges with the annotation $R_2$ (which is the typical situation when $R_1 \cap R_2$ has $n$ components that are 1D), then $P^{R_2}$ has exactly $n$ nonadjacent edges with the annotation $R_1$. For example, in Figure~\ref{fig:koreatown}, $R_1  = $ Koreatown and the annotated {polygon with holes} $P^{R_1}$ has two edges with the annotation $R_2 = $ Wilshire Center. Let $(u, v)$, with $u$ and $v$ in clockwise order, be the vertices of an edge in $P^{R_1}$ with annotation $R_2$. Because the $n$ edges are nonadjacent, {$u$ and $v$ must each} have at least $3$ neighbors (including $R_1$ {and $R_2$}). For example, in  Figure~\ref{fig:koreatown}, again consider the two edges with the annotation Wilshire Center. The two vertices $u_1$ and $v_1$ of one edge have the adjacency sequences \{Koreatown, Hancock Park, Wilshire Center\} and \{Koreatown, Wilshire Center, Little Bangladesh\}, respectively. The two vertices $u_2$ and $v_2$ of the other edge have {the} adjacency sequences \{Koreatown, Little Bangladesh, Wilshire Center\} and \{Koreatown, Wilshire Center, Pico-Union\}, respectively. For a given $(u, v)$, we seek an edge $(x,y)$ (with $x$ and $y$ in clockwise order) in $P^{R_2}$ with the annotation $R_1$ such that (1) $u$ and $y$ are annotated with the same sequences and (2) $v$ and $x$ are annotated with the same sequences. We know that there must be at least one such edge because $(u,v)$ represents a component of $R_1 \cap R_2$ and there is some edge in $P^{R_2}$ that represents the same component (so its vertices have the same sequences of adjacent regions as $u$ and $v$). In Lemma \ref{lem:unique_edge}, we prove that there is a unique such edge. If there are $n > 1$ consecutive edges $e_0, \ldots, e_{n-1}$ on the boundary of $\K_{R_1}$ with annotation $R_2$, then there are $n$ consecutive edges $e_0', \ldots, e_{n-1}'$ on the boundary of $\K_{R_2}$ with annotation $R_1$. This situation arises precisely because of the adjustments that we discuss in Appendix \ref{appendix:nbr_adjustment}. We glue $e_i$ to $e_{n-i}'$ for all $i$. If $R_1 \cap R_2$ is homeomorphic to $S^1$, then the choice of {$e_0'$} as the first edge in $P^{R_2}$ is not unique, but all choices result in topologically equivalent spaces. In Figure \ref{fig:K_koreatown}, we show the result of the gluing process for Koreatown and its neighbors.

\begin{figure}
    \centering
    \includegraphics[scale=.45]{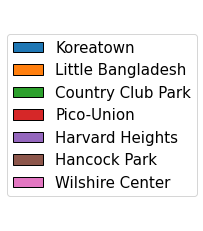}
    \subfloat[]{\includegraphics[scale=.55]{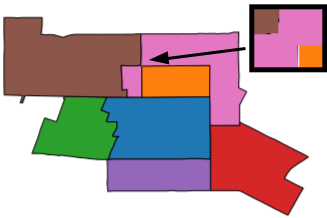}\label{fig:geo_koreatown}}
    \hspace{2mm}
    \subfloat[]{\includegraphics[scale=.33]{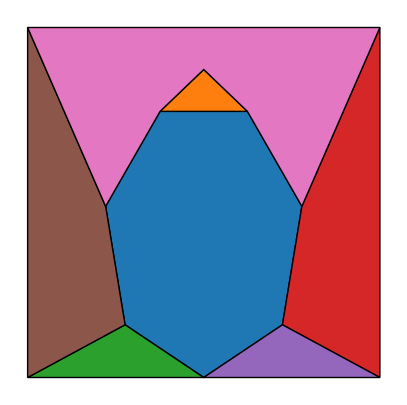}\label{fig:K_koreatown}}
    \caption{(a) A geographical set $S$ that consists of the neighborhood Koreatown and its neighbors{, as given by a {\sc shapefile} \cite{lashp}.} {Observe} that the neighborhood Little Bangladesh has only two neighbors and that the intersection between Koreatown and Wilshire Center has two components. (b) {The result of gluing Koreatown's polygon to the polygons of its neighbors.}}
    \label{fig:koreatown}
\end{figure}


\subsection{Triangulating the Polygons with Holes}

{We triangulate each {polygon with holes} $P^R$ {using} the inductive algorithm in \cite{artgallery}. We show examples of triangulated polygons with holes in Figure \ref{fig:triangulate}. The result of this triangulation process is a 2D \scom\ $\K$ with property~\ref{scom_cond}. (We assign a 2D simplex in {the} {polygon with holes} $P^R$ to the geographical region $R$.) The \scom\ $\K$ is a minimal \scom\ with property \ref{scom_cond} because (1) each {polygon with holes} $P^R$ has the minimum number of vertices and holes and (2) the number of triangles in the triangulation of $P^R$ is determined by its number of vertices and its number of holes by Euler's theorem (see \cite{artgallery}). For an example, see Figure \ref{fig:K_granada}.}

\begin{figure}
    \centering
    \subfloat[]{\includegraphics[width = .25\textwidth]{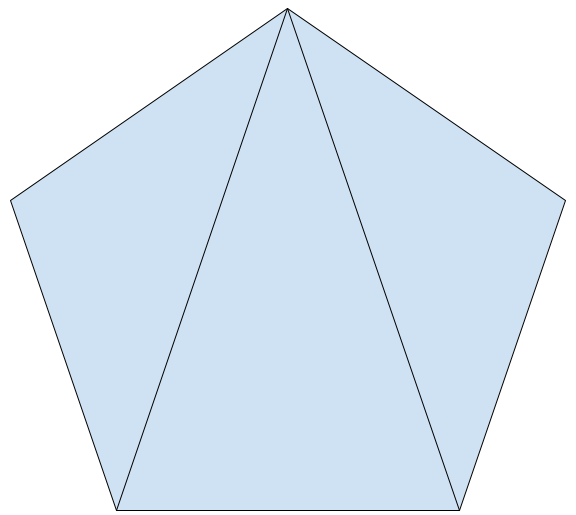}\label{fig:genus0}}
    \subfloat[]{\includegraphics[width = .25\textwidth]{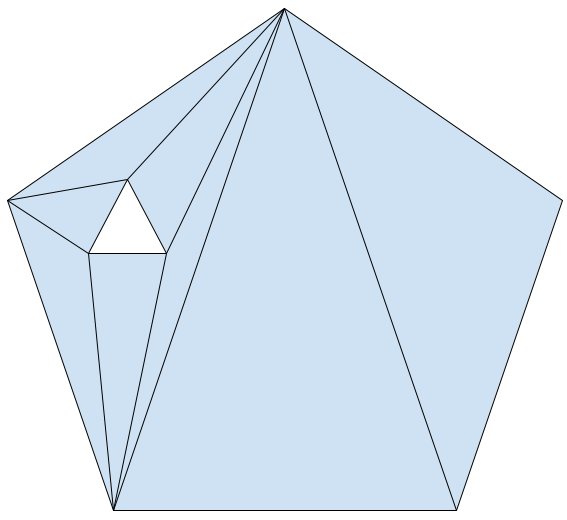}}
    \subfloat[]{\includegraphics[width = .25\textwidth]{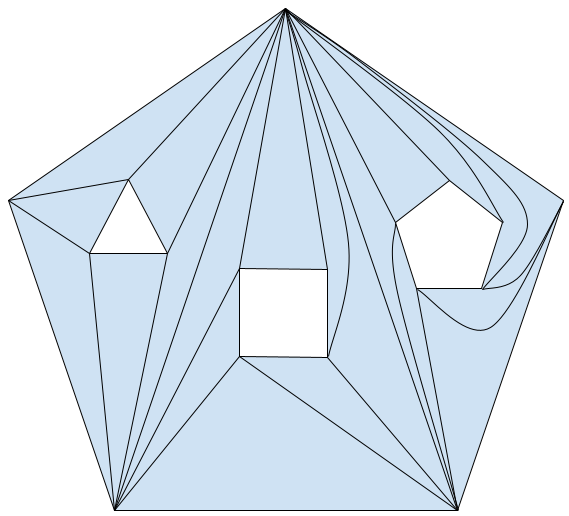}}
    \subfloat[]{\includegraphics[width=.25\textwidth]{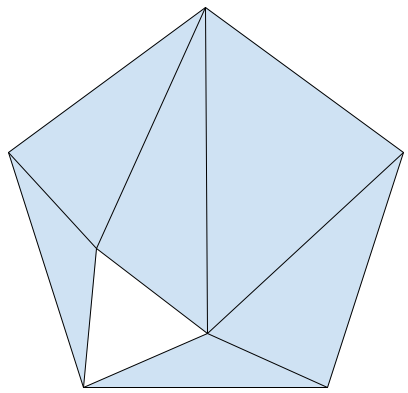}}
    \caption{Triangulation of a {polygon with holes} $P^R$ for a region $R$ when (a) $R$ has no {holes}, (b) $R$ has a single {hole}, (c) $R$ has multiple {holes},{ and (d) $R$ has a hole that touches the {exterior} boundary of $R$.}
    }
    \label{fig:triangulate}
\end{figure}


\section{Our Filtration Functions}\label{sec:filtration}

We define various filtrations that one can use with the simplicial complex $\K$ that we constructed in section \ref{scom_construction}, and we discuss how to interpret the resulting PDs and vineyards. Let $S$ be the set of geographical regions $R$ that the \scom\ $\K$ represents, and let $F: S \to \mathbb{R}$ be a real-valued function on $S$. For example, in section~\ref{sec:nyc}, $F(R)$ is the per capita {full-vaccination} rate (i.e., having received all required doses of some vaccine) for COVID-19 in NYC zip code $R$. {In sections~\ref{sec:sublevel} and \ref{sec:superlevel}, we define two} filtration functions that are induced by $F$. Given a time-dependent and real-valued function $F(t, R)$, we define time-dependent filtration functions in section~\ref{sec:timefilt}. For example, in section~\ref{sec:LA}, $F(t, R)$ is the 14-day mean per capita COVID-19 case rate in neighborhood $R$ on day $t$. From a time-dependent filtration function, we compute a vineyard.


\subsection{The Sublevel-Set Filtration}\label{sec:sublevel}

In this subsection, we define a sublevel-set filtration. In our applications, we use the 1D PH of the sublevel-set filtration to analyze local maxima in our data sets. We illustrate the idea of a sublevel-set filtration in Figure~\ref{fig:sublevel_ex}.

\begin{figure}
    \centering
    \subfloat[]{\includegraphics[width = .2\textwidth]{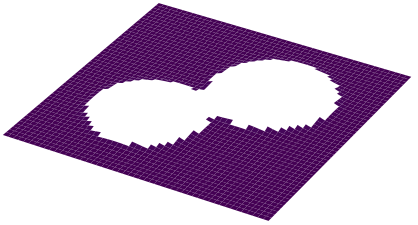}}
    \subfloat[]{\includegraphics[width = .2\textwidth]{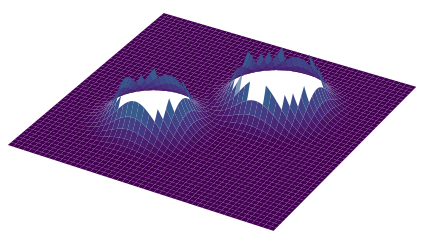}}
    \subfloat[]{\includegraphics[width = .2\textwidth]{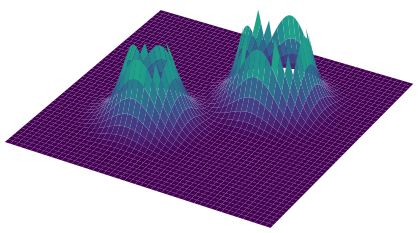}}
    \subfloat[]{\includegraphics[width = .2\textwidth]{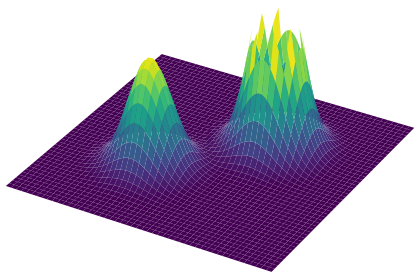}}
    \subfloat[]{\includegraphics[width = .2\textwidth]{sublevel/sublevel_sep0.png}\label{subfig:sublevel_separated}} \\
    \subfloat[]{\includegraphics[width = .2\textwidth]{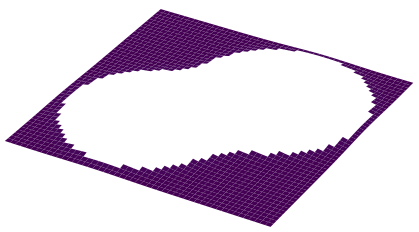}}
    \subfloat[]{\includegraphics[width = .2\textwidth]{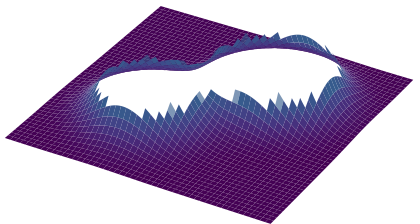}}
    \subfloat[]{\includegraphics[width = .2\textwidth]{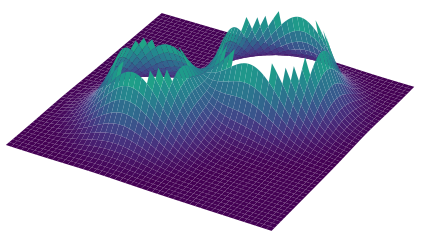}}
    \subfloat[]{\includegraphics[width = .2\textwidth]{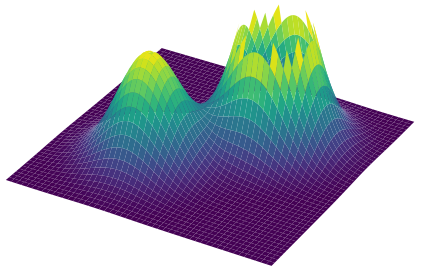}}
    \subfloat[]{\includegraphics[width = .2\textwidth]{sublevel/sublevel_nonsep0.png}\label{subfig:sublevel_nonseparated}}
    \caption{In panels (a)--(e), we show the $\alpha$-sublevel sets for increasing $\alpha$ of a function $f: \mathbb{R}^2 \to \mathbb{R}$ that has two well-separated local maxima. In (a), {for the smallest value of $\alpha$,} there is one hole that corresponds to the global maximum. In (b), a second hole appears; it corresponds to the other local maximum. In (d), the second hole is filled in. In (e), the first hole is filled in. In panels (f)--(j), we show the $\alpha$-sublevel sets for increasing $\alpha$ of a function $g: \mathbb{R}^2 \to \mathbb{R}$ whose two local maxima have the same locations and values as $f$, but are {not well-separated} from each other. {The} second hole does not appear until the sublevel set in panel (h). In all panels, the jagged edges are artifacts of the way that the Python package {\sc matplotlib} plots surfaces.}
    \label{fig:sublevel_ex}
\end{figure}

\begin{definition}[Sublevel-set filtration]\label{def:sublevel}
Let $\K$ be the \scom\ {that we obtain from our} construction in section \ref{scom_construction} for a set $S$ of regions, and let $g$ be the assignment of 2D simplices to {the} regions. Let $F: S \to \mathbb{R}$. We define the sublevel-set filtration function $f$ by considering the sublevel sets of $F$. On the 2D simplices, we define the filtration function by 
\begin{equation*}
	f(\sigma)= F(g(\sigma))\,. 
\end{equation*}
We extend the filtration function to the remaining {(lower-dimensional)} simplices {as follows. If $\sigma$ is a vertex or edge on the boundary of $\K$}, we set
\begin{equation*}
	f(\sigma) = \min_R F(R).
\end{equation*}
{Otherwise, we set}
\begin{equation}
	f(\sigma) = \min\{ f(\tilde{\sigma}) \mid {\tilde{\sigma} \text{ is a 2D simplex for which } \sigma \text{ is a vertex or edge of } \tilde{\sigma}} \}.
\end{equation}
\end{definition}

At filtration level $\alpha$, the \scom\ $\K_{\alpha}$ is the {simplicial subcomplex} of $\K$ that is induced by the union of the set of 2D simplices $\sigma$ such that $F(g(\sigma)) \leq \alpha$ and the set of vertices and edges that are on the boundary of $\K$. {Henceforth, we say that {the vertices and edges on the boundary of $\K$} are ``exterior-adjacent''.} By construction, {the underlying space of} $\K_{\alpha}$ is homeomorphic to the union of {all} regions $R$ such that $F(R) \leq \alpha$ {and} the exterior boundary. We set $f(\sigma)= \min_R F(R)$ for exterior-adjacent vertices and edges $\sigma$ for technical reasons that we will explain in a few paragraphs. In appendix \ref{sec:exterior}, we explore an alternative definition in which we set the filtration values of exterior-adjacent vertices and edges $\sigma$ to $\min_R \{ F(R) \mid R \subset C\}$, where $C$ is the connected component that contains $\sigma$.

The 1D PH of the sublevel-set filtration encodes information about the structure of the local maxima of $F$. A region $R$ {of a geographical space} is a {local maximum} if the value of $F(R)$ is larger than the value of $F(N)$ for all neighboring regions $N$ of $R$ for which $N \cap R$ is 1D. {More generally, we consider a set $E\subseteq S$ of regions {(where $|E| = 1$ is possible)} to be a \emph{local maximum} if
\begin{enumerate}
    \item the interior of $\bigcup_{R \in E}E$ is connected,
    \item the value of $F$ is constant on $E$ (we denote this value by $F(E)$), and
    \item the value of $F(E)$ is larger than the value of $F(N)$ for all regions $N \not\in E$ such that $N \cap R$ is 1D for some $R \in E$.
\end{enumerate}}
If $E$ is a local maximum, there is a 1D homology class whose death simplex is one of the simplices in the preimage $g^{-1}(E)${, where $g$ is the map from $2D$ simplices in $\K$ to geographical regions in $S$.} The class dies at filtration level $\alpha = F(E)$. For example, if $F(R)$ is the {COVID-19 case rate} in region $R$, then {the} 1D homology classes correspond to COVID-19 anomalies and the death simplex of a 1D homology class {indicates} the epicenter of that anomaly. The larger the value of $F(E)$ in comparison to {nearby} regions (including regions that {are} not necessarily immediate neighbors), the more persistent the homology class is. If the union of all regions (excluding the exterior region) is not simply connected, then there is at least one 1D homology class with an infinite death time. See Figure~\ref{fig:LAdata} for an example. The infinite 1D homology classes correspond to the holes in the {geographical space}, rather than to local maxima. The local maxima of $F$ are in one-to-one correspondence with the set of 1D homology classes with finite death times\footnote{Recall that in our definition of a local maximum, we only compare the value of a region $R$ {(or the constant value of a set {$E$ of regions})} to the {values of} neighbors $N$ that have a 1D intersection with $R$ {(or with a region in $E$)}. It is possible for two local maxima, $R_1$ and $R_2$, to have a 0D intersection. A local neighborhood of $R_1 \cup R_2$ without the union $R_1 \cup R_2$ itself is homotopy-equivalent to a figure-8, which has two 1D homology generators. One of them corresponds to $R_1$, and {the} other one corresponds to $R_2$.}. There is a canonical mapping from finite 1D homology classes to regions. A class that is represented by {the} simplex pair $(\sigma_b, \sigma_d)$ is mapped to the region $g(\sigma_d)$ that includes $\sigma_d$. The region $g(\sigma_d)$ is the location of the local maximum of $F$ that corresponds to the homology class\footnote{Let $E \subseteq S$ be the local maximum that corresponds to the 1D homology class. If $E = \{R\}$, then $g(\sigma_d) =R$. However, if $E$ contains multiple regions, then $g(\sigma_d)$ is only one of the regions {in} $E$.}, and the death simplex's filtration value $f(\sigma_d)$ is the value of the local maximum. The death simplices of the finite 1D homology classes and their filtration values give the local-maximum locations $R$ and their function values $F(R)$.

{With the 1D PH, we can do}
more than simply identify local maxima and their locations; {the 1D PH} also reveals information about {relationships {between} the local maxima.} If the local maxima are well-separated from one another, then the corresponding homology classes all have early birth times. {For example, the NYC data set has several connected components. One} can think of the global {maximum} of each connected {component as} ``totally separated'' from each other because they are on different connected components. The corresponding 1D homology classes are all born at the earliest possible filtration time, which is $\min_R F(R)$ (see Figure~\ref{fig:nyc_sublevel}). We show an example of well-separated local maxima in Figure \ref{subfig:separation}. By contrast, the two local maxima in Figure \ref{subfig:nonseparation} are not well-separated, so the homology class that corresponds to the lower peak in Figure \ref{subfig:nonseparation} is born at a {larger} filtration value than the homology class in Figure \ref{subfig:separation}. {See Figure \ref{fig:sublevel_ex} for a visualization of the {sublevel sets}.} The birth times of the 1D homology classes reflect structural information about the local maxima.

We set the filtration value of exterior-adjacent vertices and edges to {the global minimum} $\min_R F(R)$ so that 1D PH can detect local maxima on the boundary of a {geographical space}. {(We consider an alternative approach in appendix \ref{sec:exterior}.)} This is important for the LA data set of COVID-19 case rates. In Figure~\ref{fig:la_vy_map}, we observe that many of the {most-persistent} COVID-19 anomalies are on the boundary of the {geographical space}{; it} is crucial that we are able to detect them. If we had {not defined the exterior-adjacent filtration values {in} this way, then} the filtration value of exterior-adjacent vertices and edges $\sigma$ would {be} $F(R)$, where $R$ is the unique region that is adjacent to $\sigma$. If $R$ is a local maximum, its corresponding 1D homology class {is born and dies} at filtration level $\alpha = F(R)$. In the PD, {it then appears} as a point on the diagonal. Therefore, for 1D PH to detect local maxima on the boundary of a {geographical space}, we must adjust the filtration values of exterior-adjacent vertices and edges.

The 0D homology classes correspond to local minima of $F$. However, unlike for the 1D homology classes, there is not a natural mapping from {0D} homology classes to the locations of the minima. In appendix \ref{sec:0D}, we discuss the interpretation and computation of 0D homology classes.


\subsection{The Superlevel-Set Filtration}\label{sec:superlevel}

An alternative to using the sublevel-set filtration from section~\ref{sec:sublevel} is {to use superlevel sets of $F$} to construct a superlevel-set filtration. In our {case study on COVID-19 vaccination rates in NYC}, we use {a} superlevel-set filtration to analyze local minima {of} {the vaccination rate}. {We define a \emph{local minimum} analogously to the way {that} we defined a local maximum in section \ref{sec:sublevel}.} We illustrate the idea of the superlevel-set filtration in Figure~\ref{fig:superlevel_ex}.

\begin{figure}
    \centering
    \includegraphics[width = .24\textwidth, valign = t]{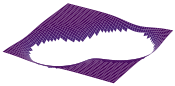}
    \includegraphics[width = .24\textwidth, valign = t]{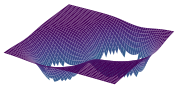}
    \includegraphics[width = .24\textwidth, valign = t]{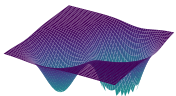}
    \includegraphics[width = .24\textwidth, valign = t]{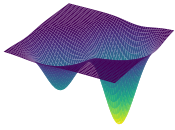}
    \caption{The $\alpha$-superlevel sets, with $\alpha$ decreasing from left to right, for the graph of a function $f: \mathbb{R}^2 \to \mathbb{R}$ with two local minima.}
    \label{fig:superlevel_ex}
\end{figure}

\begin{definition}[Superlevel-Set Filtration]\label{def:superlevel}

Let $F: S \to \mathbb{R}$ for a set $S$ of regions. The superlevel-set filtration function $f$ is the sublevel-set filtration function that is induced by $-F$.
\end{definition}

At filtration level $-\alpha$, the simplicial complex $\K_{-\alpha}$ is the {simplicial subcomplex} of $\K$ that is induced by the union of {the set of exterior-adjacent simplices and} the set of 2D simplices $\sigma$ for which $F(g(\sigma)) \geq \alpha$. By construction, {the underlying space of} $\K_{-\alpha}$ is homeomorphic to the union of regions $R$ for which $F(R) \geq \alpha$ {along with the exterior boundary}. Local maxima of $F$ now correspond to 0D homology classes, and local minima of $F$ now correspond to 1D homology classes; this is the opposite situation from the sublevel-set filtration. Our discussion of local maxima for the sublevel-set filtration in section \ref{sec:sublevel} applies to local minima for the superlevel-set filtration, and our discussion of local minima for the sublevel-set filtration in section \ref{sec:sublevel} applies to local maxima for the superlevel-set filtration. The only difference is that the filtration values in the superlevel-set filtration are the additive inverses of the function values of $F$. This implies, for example, that the death filtration value of a 1D homology class that corresponds to a local minimum at region $R$ is $\alpha = -F(R)$, rather than $\alpha = F(R)$.


\subsection{A Time-Dependent Filtration}\label{sec:timefilt}

Suppose that we have a time-dependent, real-valued function $F(t, R)$ whose domain is $\{t_0, t_1, \ldots, t_n\} \times S$, where $t_0 \in \mathbb{R}$ is the initial time and $t_n \in \mathbb{R}$ is the final time. For example, in section~\ref{sec:LA}, the value of $F(t, R)$ is the 14-day mean per capita COVID-19 case rate in Los Angeles on day $t$. We seek to analyze the structure of local extrema as they change {with} time.

\begin{definition}[Time-Dependent Sublevel-Set Filtration]
Let $F:\{t_0, t_1, \ldots, t_n\} \times S \to \mathbb{R}$ be a time-dependent function on a set $S$ of regions{, and let $\K$ be the \scom\ {for $S$} from the construction in section \ref{scom_construction}.} At each time ${t_i} \in \{t_0, t_1, \ldots, t_n\}$, we define the time-dependent filtration function $f({t_i}, \cdot)$ to be the sublevel-set filtration that is induced by $F({t_i}, \cdot)$. To extend this filtration function to the entire interval $[t_0, t_n]$, we linearly interpolate {$f(\cdot, \sigma)$ on each subinterval $[t_i, t_{i+1}]$ {for} all simplices $\sigma \in \K$.}
\end{definition}

In the present paper, we only use the time-dependent sublevel-set filtration, but one can analogously define a time-dependent superlevel-set filtration. We have implemented both of these filtrations in our code.

We use a time-dependent sublevel-set filtration to construct a vineyard. This allows us to track how the extrema move in both space and time. As in section~\ref{sec:sublevel}, each finite vine corresponds to a local maximum whose location at time $t$ is given by the region $g(\sigma_d(t))$ that contains the vine's time-dependent death simplex $\sigma_d(t)$\footnote{It is known that vineyards are not stable \cite{crocker}. A small perturbation in filtration values can cause crossing of vines that previously did not cross (i.e., it is an ``avoided crossing''). This, in turn, causes simplex pairings to change. Therefore, the geographical region $g(\sigma_d(t))$ that corresponds to a particular vine at time $t$ is sensitive to small perturbations in filtration values.
}.
The length of a vine corresponds to its persistence in time.


\section{Case Studies}\label{sec:apps}

We now apply our methods to two data sets, which we {illustrate} in Figure~\ref{fig:datasets}.

\begin{figure}
    \centering
    \subfloat[NYC zip codes]{\includegraphics[height = 4.4cm]{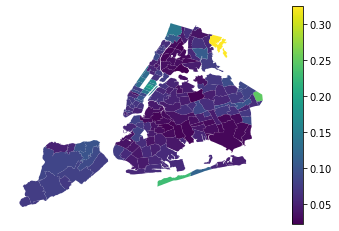}\label{fig:NYCdata}} 
    \subfloat[LA neighborhoods]{\includegraphics[height = 4.4cm]{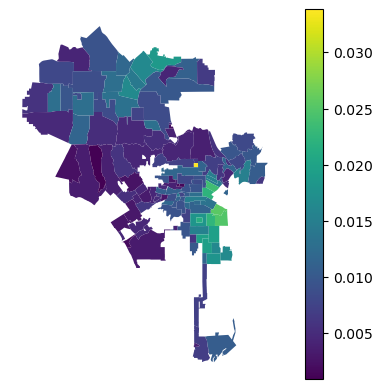}\label{fig:LAdata}}
    \caption{We show
    (a) {the} per capita {COVID-19 full vaccination} rate in New York City (NYC) by (modified) zip code on 23 February 2021 and
    (b) {the} 14-day mean per capita COVID-19 case rate in {the city of} Los Angeles (LA) by neighborhood on 30 June 2020. In {both} {(a) and (b)}, the white regions are geographical regions that do not belong to the depicted city.
    }
    \label{fig:datasets}
\end{figure}


\subsection{COVID-19 Vaccination Rates in New York City}\label{sec:nyc}

We examine vaccination rates in (modified) zip codes of NYC\footnote{{The NYC Department of Health \& Mental Hygiene uses modified zip-code tabulation areas (MODZTCA) for} COVID-19 data \cite{NYCshp}. In these modified zip codes, some zip codes with small populations are combined \cite{modzcta}. We henceforth refer to modified zip codes as simply ``zip codes''.}. We demonstrate the effects of the two filtrations that we defined in section \ref{sec:filtration}. The geographical boundaries of the zip codes are given by a {\sc shapefile} \cite{NYCshp}. From {the {\sc shapefile}}, we construct a \scom\ $\K$ in the manner that we described in section \ref{scom_construction}. The {vaccination} data set, which we obtained from the NYC Department of Health \& Mental Hygiene website \cite{NYCvax}, consists of the number of fully vaccinated people in each zip code on 23 February 2021\footnote{{At the time, the} NYC Department of Health \& Mental Hygiene {defined} ``fully vaccinated'' people to be individuals who {either had} received both doses of the Pfizer or Moderna vaccine or {had received} one dose of the Johnson \& Johnson vaccine. (This differs from common parlance at that time, in which people were sometimes {considered to be} ``fully vaccinated'' only after two weeks had passed since their final dose of a vaccine.)}. For each zip code, we divide this number by its population estimate in \cite{NYCvax} to obtain a per capita vaccination rate. For zip code $R$, we define $F(R)$ to be the per capita vaccination rate in $R$ on 23 February 2021.

We do not possess the {daily} vaccination-rate data that is necessary to compute a vineyard, so instead we calculate the PH of $\K$ with the {sublevel-set} and superlevel-set filtrations from sections~\ref{sec:sublevel} and~\ref{sec:superlevel}. We show the resulting PDs for the 1D PH in Figure~\ref{fig:nycPD}. As we described in section~\ref{sec:sublevel}, the points in the PD {from the sublevel-set filtration} correspond to zip codes in which vaccination rates are higher than in the {neighboring} zip codes. The death filtration level of a homology class is the vaccination rate in that zip code, and the birth filtration level of a homology class reflects the extent of spatial isolation of that zip code from other local {maxima. An} earlier birth filtration implies more spatial isolation. Similarly, the points in the superlevel-set filtration {PD} correspond to zip codes in which the vaccination rate is lower than in {the neighboring} areas. As we discussed in section~\ref{sec:sublevel}, we obtain the zip code {that is associated with a} homology class from its death simplex $\sigma_d$. We color {the} points in the PDs by {the boroughs of their corresponding zip codes}.

\begin{figure}
    \centering
    \subfloat[Sublevel-set filtration]{\includegraphics[width = .5\textwidth]{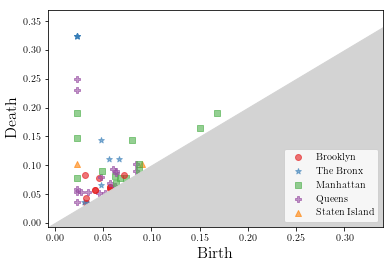}\label{fig:nyc_sublevel}}
    \subfloat[Superlevel-set filtration]{\includegraphics[width = .5\textwidth]{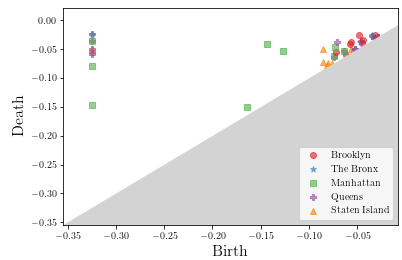}\label{fig:nyc_superlevel}} 
    \caption{PDs for the 1D PH of the NYC simplicial complex with filtrations that are induced by the per capita full vaccination rate by zip code on 23 February 2021. We show only the finite homology classes. Each point in a PD corresponds to a zip code, which we label according to its borough \cite{boroughs}, that has (a) a higher vaccination rate than its neighboring zip codes or (b) a lower vaccination rate than its neighboring zip codes.}
    \label{fig:nycPD}
\end{figure}

In Figures \ref{fig:nyc_maxima} and \ref{fig:nyc_minima}, we highlight the locations of the maxima and minima, respectively. In Figures \ref{subfig:maxima_death} and \ref{subfig:minima_death}, {we color the extrema {based on their vaccination rates}}. In Figure \ref{subfig:minima_death}, we observe that the minima all have {near-$0$} vaccination rates. In Figures \ref{subfig:maxima_pers} and \ref{subfig:minima_pers}, {we color {each zip code} according to the} persistence {(i.e., {the value} $\text{death} - \text{birth}$)} of {its} corresponding homology class. These two figures incorporate global information about the structure of the extrema, as {we} described {in the paragraph above and in section \ref{sec:filtration}}. For example, in Figure \ref{subfig:minima_pers}, we observe that some of the minima ({specifically,} those with {the} {largest values of persistence}) are significantly more spatially separated than others, even though all {of} the minima have similar vaccination rates. {A larger} persistence of a local minimum indicates a greater difference in the vaccination rate between the minimum and the {neighboring} zip codes. A zip code that is a local minimum with {a larger} persistence may {have a} greater inequity in vaccine {access than {its}} neighboring regions. {Such insights may be useful for sociologists and policy makers.}

\begin{figure}
    \centering
    \subfloat[]{\includegraphics[width=.4\textwidth]{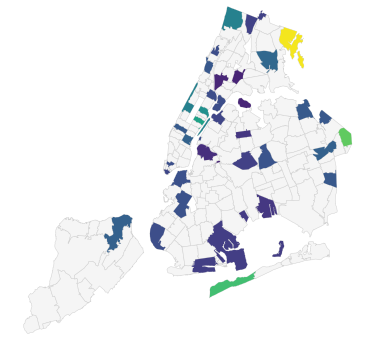}\label{subfig:maxima_death}}
    \subfloat[]{\includegraphics[width=.4\textwidth]{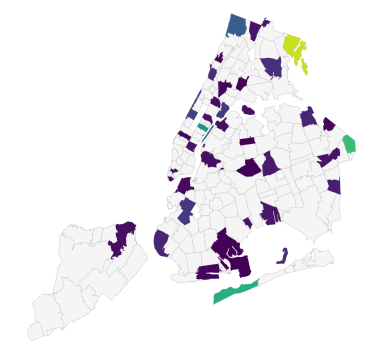}\label{subfig:maxima_pers}}
    \hspace{3mm}
    \subfloat{\includegraphics[scale=.3]{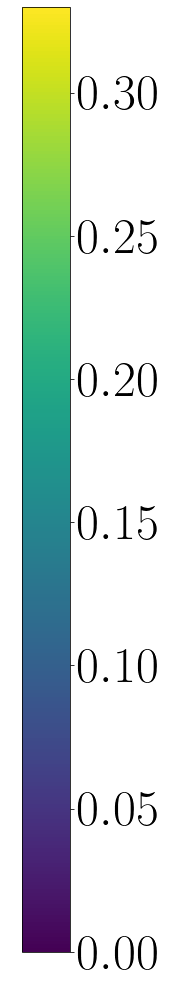}}
    \caption{Maps of the local maxima of the NYC vaccination-rate function. (a) Color corresponds to the vaccination rate of {a} zip code. (b) Color corresponds to the persistence (i.e., $\text{death} - \text{birth}$) of the corresponding homology class.}
    \label{fig:nyc_maxima}
\end{figure}

\begin{figure}
    \centering
    \subfloat[]{\includegraphics[width=.4\textwidth]{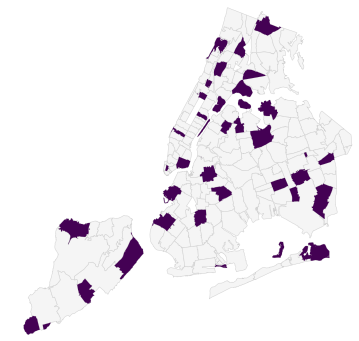}\label{subfig:minima_death}}
    \subfloat[]{\includegraphics[width=.4\textwidth]{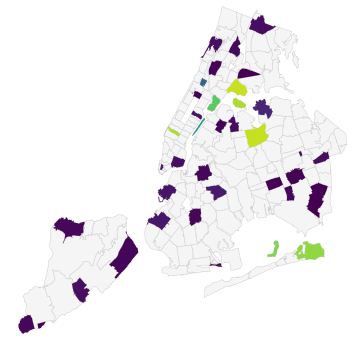}\label{subfig:minima_pers}}
    \hspace{3mm}
    \subfloat{\includegraphics[scale=.3]{nyc/colorbar.png}}
    \caption{Maps of the local minima of the NYC vaccination-rate function. (a) Color corresponds to the vaccination rate of {a} zip code. (b) Color corresponds to the persistence (i.e., $\text{death} - \text{birth}$) of the corresponding homology class.}
    \label{fig:nyc_minima}
\end{figure}

An issue arises from the fact that several of the NYC zip codes are islands and {thus are isolated}. These islands are trivial extrema because they are not adjacent to any other zip codes. One may wish to exclude these trivial extrema from {a} PD. In appendix \ref{sec:exterior}, we propose alternative methods for handling disconnected geographical spaces such as NYC.

One can use the PDs in Figure \ref{fig:nycPD} to study inequities in vaccine access. For example, one may seek to discern patterns in demographic data that correspond to the {most-persistent} points in the PDs. For {interested readers, we provide} some demographic data in appendix \ref{sec:demographic}.


\subsection{COVID-19 Case Rate in the City of Los Angeles}\label{sec:LA}

We now examine {time-dependent} COVID-19 case rates in neighborhoods of the city of Los Angeles (LA)\footnote{We exclude Angeles National Forest because it {has only} 20 inhabitants.}. The geographical boundaries of the neighborhoods are given by a {\sc shapefile} \cite{lashp}. From the {\sc shapefile}, we construct a simplicial complex $\K$ in the manner that we described in section~\ref{scom_construction}. We also know the number of cases in each neighborhood on each day from 25 April 2020 {to} 25 April 2021. For each neighborhood, we divide the case count by the neighborhood population to obtain per capita case rates, and we calculate a running 14-day mean\footnote{On day $t$, we take the mean of the case rates on days $t$, $t-1$, \ldots, $t-13$. Some outlets (e.g., \cite{stat_covidtracker}) report running 14-day means of COVID-19 case counts, and other outlets (e.g., \cite{nytimes_covidtracker}) report 14-day trends.} on each day to smooth the data. For neighborhood $R$ and time $t \in \{0, {1,} \ldots, 365\}$, we define $F(t, R)$ to be the 14-day mean per capita case rate in $R$ on day $t$ after $25$ April 2020. We compute the vineyard for a simplicial complex $\K$ using the time-dependent sublevel-set filtration that is induced by $F(t, R)$. {We show the most important and interesting subsets of our vineyard in Figures \ref{fig:la_vineyards_new} and \ref{fig:la_vy_may}. See Figure \ref{fig:la_vy_full} for the full vineyard.}

\begin{figure}
    \centering
    \subfloat[]{\includegraphics[height = .33\paperheight]{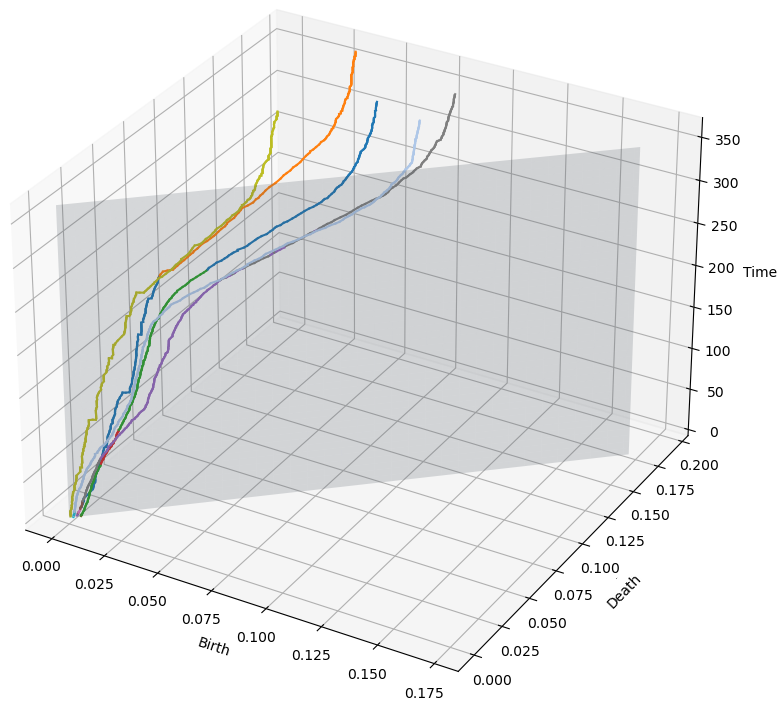}}\\
    \subfloat[]{\includegraphics[height=.33\paperheight]{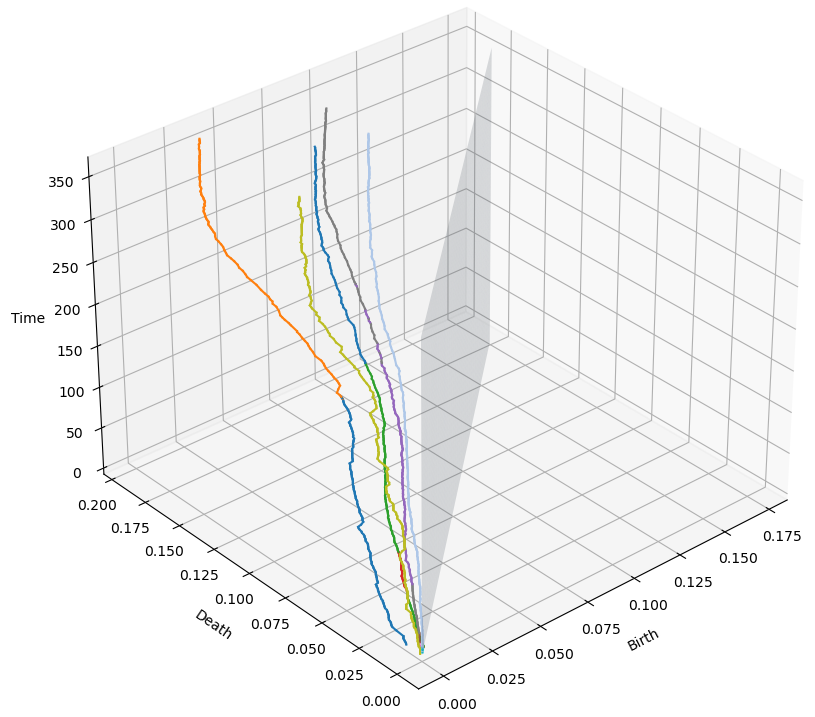}}
    \caption{(a) The five most-persistent vines of the vineyard for the LA \scom\ with a sublevel-set filtration from the 14-day mean per capita case rate during the period 25 April 2020--25 April 2021. (See Figure \ref{fig:la_vy_full} for the full vineyard.) Each vine {corresponds to} a COVID-19 anomaly. We color each vine according to the geographical locations of its associated anomaly. Because the geographical location of an anomaly can change {with} time, a single vine can have multiple colors. {(See Figure~\ref{fig:la_vineyards_new_legend} for the legend.)} (b) A different view of the same five vines.}
    \label{fig:la_vineyards_new}
\end{figure}

\begin{figure}
    \centering
    \includegraphics[width = \textwidth]{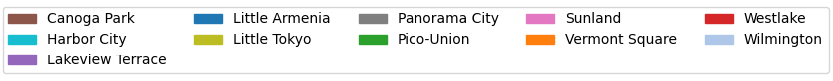}
    \caption{The legend for Figure~\ref{fig:la_vineyards_new}. Each of the depicted regions is a local maximum of the COVID-19 {case-rate function} for some subset of the time period 25 April 2020--25 April 2021.}
    \label{fig:la_vineyards_new_legend}
\end{figure}

The vines in the vineyard correspond to COVID-19 anomalies, which we define to be neighborhoods that have a higher running 14-day mean COVID-19 case rate than the surrounding neighborhoods for at least one day. Anomalies that are more spatially isolated yield vines with {earlier} birth-filtration levels, and anomalies with high case rates yield vines with late death-filtration levels. See section~\ref{sec:sublevel} for a detailed discussion. We color each vine according to the geographical location(s) of its anomaly. As we discussed in section~\ref{sec:timefilt}, we obtain the anomaly location(s) from the time-dependent death simplex $\sigma_d(t)$ of a vine. The function $\sigma_d(t)$ is a piecewise-constant function; as it changes, so does the location of the associated anomaly. Therefore, the color of a vine can change {with} time. For example, consider Figure~\ref{fig:la_vineyards_new}, where we show the five most-persistent vines\footnote{We {defined} the {persistence} of a vine in section \ref{sec:vineyard}.}. The global maximum of the data set is initially in Little Armenia, but it moves to Vermont Square at about $t = 220$. In the vineyard, we see this from the vine that is initially blue (for Little Armenia) from time $t = 0$ until about $t = 220$ and then orange (for Vermont Square) starting from about time $t = 220$ through time $t = 365$. There are also other vines whose locations change {with} time. Such geographical location changes do not need to be adjacent, but they often are near each other. In Figure~\ref{fig:la_vy_map}, we highlight these anomalies on a map.

\begin{figure}
    \centering
    \includegraphics[scale=.5]{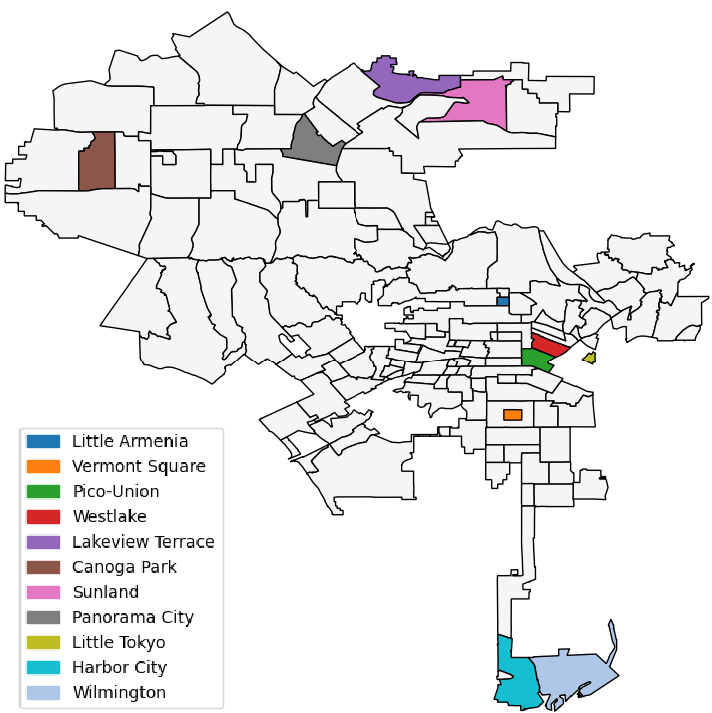}
\caption{A map of the {most-persistent} anomalies of the COVID case-rate function in LA during the time period 25 April 2020--25 April 2021. Each of the highlighted regions is a local maximum of the COVID-19 case-rate function for some subset of the time period.}
    \label{fig:la_vy_map}
\end{figure}

A vineyard encodes the temporal persistence of anomalies. The length of time that a vine is not on the diagonal plane of a vineyard, which we henceforth call the ``length'' of a vine, is the amount of time that an anomaly exists {in the vineyard}. At the beginning of the COVID-19 pandemic, all neighborhoods had low per capita case rates. We expect {an emerging anomaly} to have a low case rate for a long time and then for the case rate to grow rapidly starting at some later time. An emerging anomaly in the {``low-case-rate''} phase yields a vine that is close to the diagonal for a long time. By examining the lengths of vines, we hypothesize that one can distinguish between concerning emerging anomalies (i.e., those that may become major COVID-19 anomalies in the future) and anomalies of lesser concern, even when the anomalies have similar case rates.

In Figure~\ref{fig:la_vy_may}, we show case rates early in the time period that we track {(and close to the ``beginning''\footnote{The COVID-19 pandemic was declared a national emergency on 13 March 2020 \cite{national_emergency}, and the city of LA closed its public schools and ordered the closure of restaurants, bars, and gyms on 16 March 2020 \cite{LA_closures}.} of the COVID-19 pandemic) by computing the vineyard for the period 25 April 2020--25 May 2020.} In {the depicted} vineyard, we exclude the twenty most-persistent vines to more easily {see} the vines that are close to the diagonal plane. Many of these latter vines are short, so their associated anomalies are short-lived. The longer vines are anomalies that are longer-lived and thus of greater concern in the long run, even though they are close to the diagonal during the period 25 April 2020--25 May 2020. For example, there is an anomaly at Wilmington that we show with the light-blue vine. This vine is close to the diagonal plane, but it has {a} {large} temporal persistence during the period 25 April 2020--25 May 2020. In Figure~\ref{fig:la_vineyards_new}, we see that Wilmington eventually becomes one of the {most-persistent anomalies} in LA.

\begin{figure}
    \centering
    \subfloat[]{\includegraphics[height = .345\paperheight]{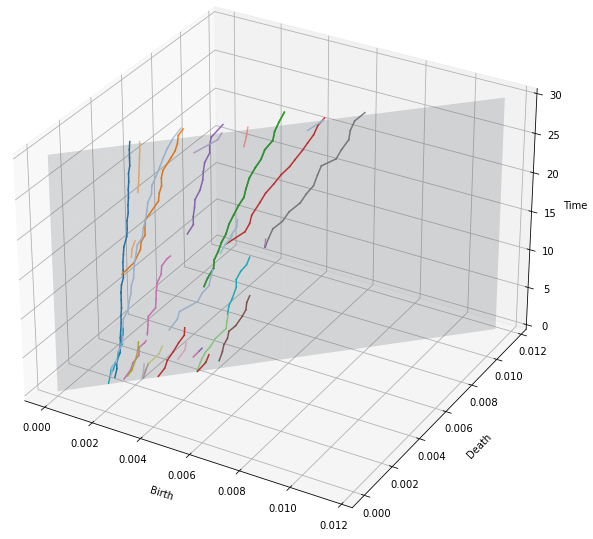}}\\
    \subfloat[]{\includegraphics[height=.345\paperheight]{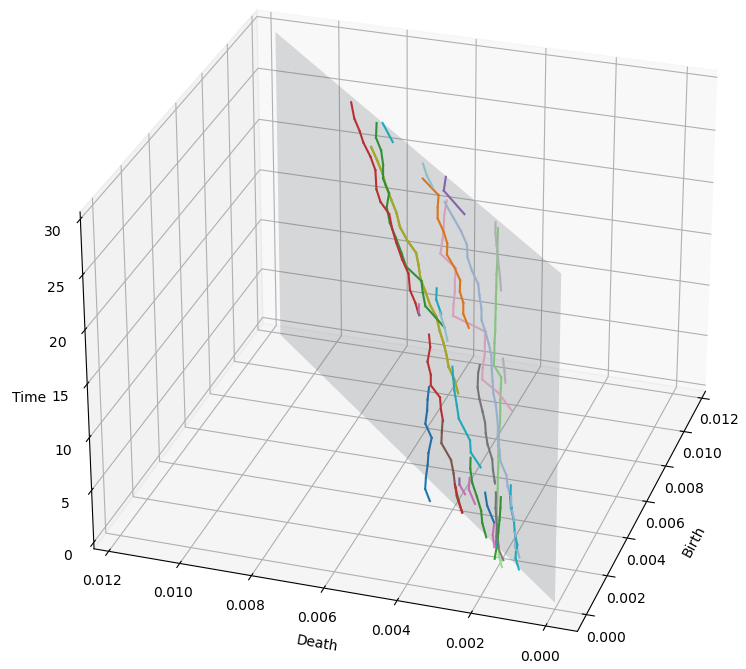}}
    \caption{{(a)} Vineyard for the LA \scom\ with a sublevel-set filtration {for the} 14-day mean per capita case rate during the period 25 April 2020--25 May 2020. We exclude the 20 most-persistent vines to more easily see the vines {that are} near the diagonal plane. Each vine {is associated with} a COVID-19 anomaly{, and} we color each vine according to the geographical location(s) of its anomaly. {See Figure~\ref{fig:may_vy_legend} for the legend.} {(b) A different view of the same set of vines.}}
    \label{fig:la_vy_may}
\end{figure}

\begin{figure}
    \centering
    \includegraphics[width = \textwidth]{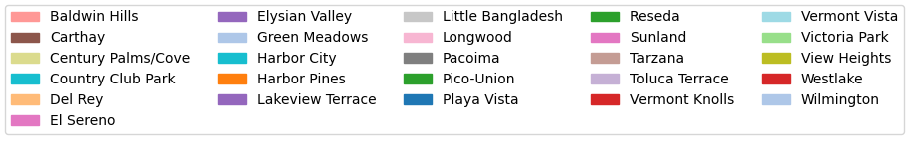}
    \caption{{The legend for Figure~\ref{fig:la_vy_may}. Each of the depicted regions is a local maximum of the COVID-19 {case-rate function} for some subset of the time period 25 April 2020--25 May 2020.}}
    \label{fig:may_vy_legend}
\end{figure}

\section{Discussion}\label{sec:discussion} 

In our approach, we needed to make a variety of choices. There are other ways to construct a simplicial complex to represent a geographical space. There are also other choices of topological tools for analyzing time-varying data. We briefly discuss some of these possibilities in the next several paragraphs.

{If one only cares about local information ({specifically,} the locations and values of the extrema) and not {about} global information ({such as the} spatial separation between extrema), then an alternative method for constructing a \scom\ $\K$ is to construct the dual graph of the set of regions $S$. That is, for each region component $R$, there is a vertex $v_R \in \K_R$, and if regions $R_1$ and $R_2$ are adjacent, then there is an edge between $v_{R_1}$ and $v_{R_2}$. If we wish to study local maxima of a function $F: S \to \mathbb{R}$, then we define the filtration of an edge $e = (v_{R_1}, v_{R_2})$ to be $f(e) = \max\{ F(R_1), F(R_2)\}$ and we define the filtration of a vertex $v_R$ to be $f(v_R) = 0$. (There is an analogous definition for studying local minima.) In the 0D PH of the FSC $(\K, f)$, the homology classes correspond to local maxima. If a homology class's birth simplex is the vertex $v_R$, then $R$ is the corresponding local maximum and $F(R)$ is the death filtration level of the homology class. All the 0D homology classes are born at $0$; thus the birth filtration level does not provide any additional information, as it did for our construction in section \ref{scom_construction}. Because of this, we do not obtain any global information from the $PH$ of $(\K, f)$.
}

Rasterization gives {another} alternative method to construct a simplicial complex from {\sc shapefile} data. When one rasterizes a {\sc shapefile}, one can transform the resulting image into a \scom\ by imposing the pixels of the image onto a triangulation of the plane. However, our approach has several key advantages over rasterization. First, the number of simplices in the \scom\ that one obtains by rasterizing a {\sc shapefile} is orders-of-magnitude larger than the number of simplices in our construction. Computing the PH of a \scom\ with fewer simplices allows significantly faster computations. Second, the \scom\ that one obtains by rasterization has no guarantee of ``topological correctness'', as property~\ref{scom_cond} may not hold. The extent to which the resulting \scom\ is topologically correct depends on the resolution of the rasterization, and using a higher resolution requires more simplices. Our construction of simplicial complexes also yields a natural way to map a 2D simplex to the geographical region that contains it. We use this preservation of geographical information to find the locations of the local extrema. Lastly, our construction allows us to {detect anomalies on the boundary of a geographical space.}

Our construction uses {geographical adjacencies}, but one may instead wish to employ ``effective'' distances between regions. One can calculate effective distances using mobility and transportation data. Two regions that are closely connected via transportation are effectively closer than they are based on direct geographical considerations; this affects phenomena such as the dynamics of infectious diseases \cite{science_eff, india_eff}.

We used only 1D PH to study extrema, but one can alternatively use 0D PH if one is not interested in the geographical locations of the extrema; we discuss this in appendix \ref{sec:0D}. In appendix \ref{sec:exterior}, we discuss alternative filtrations that one can apply to geographical spaces (such as NYC) that are disconnected. We used a time-dependent function on a geographical space to compute vineyards, but an alternative is to use an approach that is based on multiparameter PH. In appendix \ref{sec:multi}, we discuss how one does this when the time-dependent function $F(\cdot, R)$ is monotonic for all regions $R$. When $F(\cdot, R)$ is not monotonic for all $R$, we discuss in appendix \ref{sec:zigzag} how one can use an approach that is based on multiparameter zigzag PH. Both multiparameter PH and multiparameter zigzag PH are difficult to visualize, and they both suffer from a lack of easily interpretable invariants. Consequently, we chose to compute vineyards for our applications.


\section{Conclusions}\label{sec:conclusion}

We developed methods to directly incorporate spatial structure into applications of topological data analysis ({specifically,} of persistent homology) to geospatiotemporal and geospatial data. We defined a way to construct a simplicial complex that efficiently and accurately represents a geographical space. Given a function on a geographical space, we defined filtration functions on a simplicial complex such that the homology classes are in one-to-one correspondence with either local minima or local maxima. By constructing a vineyard, one can track how the local extrema move and change {with} time.

We conducted case studies using COVID-19 vaccination and case-rate data. In one case study, we examined {geospatial} {vaccination-rate} structure in New York City on one day. In our other case study, in which we examined geospatiotemporal data, we constructed a vineyard to {examine} COVID-19 case-rate anomalies in the city of Los Angeles over the course of one year. From the vineyard, we identified the locations of these anomalies and measured the severity of the {associated} disease outbreaks. The vineyard also {captures} information about the relationships between anomalies, such as the extent to which they are isolated from each other. We calculated the temporal persistence of {each} anomaly based on the length of its corresponding vine.

There are several ways to build on our research. It is desirable to discover how to use a vineyard to produce systematic forecasts of how a disease (or something else) will spread in space and time. We hypothesized in section~\ref{sec:LA} that one can identify ``emerging anomalies'' in the COVID-19 data set as vines that are long but close to the diagonal plane. In other applications, one may wish to forecast which locations of local extrema will have the largest data values and/or {the largest} temporal persistences. One may also {wish} to forecast how {extrema} will move in space. It will be valuable to investigate how to use the output of our approach as an input to forecasting models.

Our approach is useful for a wide variety of applications, and it seems possible to generalize it for many others. For example, given spatiotemporal voting data, one can identify regions that vote differently from the neighboring regions. This would allow one to generalize the work of \cite{feng2021} to track the intensity of voting differences and study spatial relationships between different political islands. Our methodology is not restricted to geographical data. {It} is applicable whenever one has a surface that is partitioned into a finite number of regions and a real-valued function (or a sequence of real-valued functions) on those regions. For example, it may be possible to apply our approach to grayscale image data by partitioning an image into regions in which pixel values are close to each other. It also seems possible to extend our approach to higher dimensions; this would require constructing a higher-dimensional simplicial complex when one has adjacency information for the higher-dimensional regions. For example, in three dimensions, one can use such an extension of our approach to study atmospheric, oceanic, and video dynamics.


\appendix
\section{Details of our Simplicial-Complex Construction}\label{sec:details}

\subsection{Boundary-Sequence Adjustment}\label{appendix:nbr_adjustment}

Before constructing the {polygons with holes} $P^R$ for each region $R$, we adjust the boundary sequences if necessary. The adjustment procedure proceeds as follows. Let $D_0^R, D_1^R, \ldots, D_{h_R}^R$ be the disks in the statement of assumption \ref{geo_bdrycond}, let $B_i^R = \partial D_i^R$, and let $S^R_i$ denote the sequences of neighbors around $B_i^R$. First, we adjust the sequences so that, for each region $R$ and each $B_i^R$, the first element of $S_i^R$ has a 1D intersection with $R$. We then adjust the sequences {so} that $|S_i^R| \geq 3$ for all $R$ and $i$. {When} $|S_i^R| < 3$, there are two cases:
\begin{enumerate}
\item {(Case 1)} If $|S_i^R| = 1$, let $N$ be the unique element of $S_i^R$. This situation occurs if $R$ is an island, and it can also occur if $R$ lies inside $N$ or if $N$ lies inside $R$. We adjust $S_i^R$ to be the sequence $\{N, N, N\}$. If $N$ is not the exterior region, let $j$ be the index such that $B_j^N$ intersects $R$. Adjust $S_j^N$ to be the sequence $\{R, R, R\}$ to compensate for the adjustment that we made to $S_i^R$.

\item {(Case 2)} If $|S_i^R| = 2$, let $N_1$ and $N_2$ be the two elements of $S_i^R$. If $B_i^R$ intersects $R$, then {$R$ is adjacent to the exterior; without loss of generality, let $N_1$ denote the exterior region.} For example, in Figure~\ref{fig:geo_koreatown}, $S_0^{\text{Little Bangladesh}} = \{\text{Koreatown}, \text{Wilshire Center}\}$. We adjust $S_i^R$ to be the sequence $\{N_1, N_1, N_2\}$. If $N_1$ is not the exterior region, which occurs if $R$ is not adjacent to the exterior, then we also adjust $S_j^{N_1}$ to compensate, where $j$ is the index of the boundary component of $N_1$ that intersects $R$. In this case, we adjust $S_j^{N_1}$ by repeating $R$ an additional time. 
\end{enumerate}


\subsection{Construction of $\K$ from the Set
$\{P^R \mid R \in S\}$}\label{sec:lemmas}

We present two lemmas that we used in section \ref{scom_construction} to construct $\K$ by gluing together the set $\{P^R \mid R \in S\}$ of polygons with holes.

\begin{lemma}\label{lem:enclosure}
Let $R_1$ and $R_2$ be connected regions in a set $S$ that
satisfies assumptions \ref{geo_firstcond}--\ref{geo_lastcond}. Let $D_0, \ldots, D_h$ be the disks in the statement of \ref{geo_bdrycond} for $R_1$. {It is then the case that exactly one of the following statements} is true:
\begin{enumerate}
    \item $R_2 \subseteq \intr(D_0)^c$ and $R_2 \cap \intr(D_i) = \emptyset$ for all $i > 0${; or}
    \item {there} is an $i > 0$ such that $R_2$ is enclosed in $D_i$ and $R_2 \cap \intr(D_j) = \emptyset$ for all $j \neq i$.
\end{enumerate}
\end{lemma}

\begin{proof}
Because the interiors of $R_1$ and $R_2$ do not intersect, {it must be true} that $\intr(R_2) \subseteq \intr(D_0)^c \cup \Big(\bigcup_{i=1}^h \intr(D_i)\Big)$. Therefore, 
\begin{equation*}
    \intr(R_2) = \Big(\intr(D_0)^c \cap \intr(R_2)\Big) \cup \Big( \bigcup_{i=1}^h \intr(D_i) \cap \intr(R_2)\Big){\,.}
\end{equation*}    
The {claim} follows because $\intr(R_2)$ is connected and $\intr(D_0)^c, \intr(D_1), \ldots, \intr(D_h)$ are {pairwise disjoint}.
\end{proof}

\begin{lemma}\label{lem:unique_edge}

Let {$P^R$ be the annotated polygon with holes} for a connected region $R$, let $v$ be a vertex in $P^R$, and let $\{R, N_1, \ldots, N_n\}$ be the sequence of region adjacencies for $v$. If $n \geq 2$ and $N_1, \ldots, N_n$ are connected, then $P^R$ has at most one other vertex $w$ with the same set of region adjacencies. Additionally, if $w$ exists, its sequence of region adjacencies must be $\{R, N_n, \ldots, N_1\}$, which is the mirror of the orientation of neighbors around $v$.
\end{lemma}

\begin{proof}
{
Suppose that $w$ is another vertex in $P^R$ with the same set of region adjacencies as $v$. Let $v'$ and $w'$ denote the points on the boundary of $R$ that correspond, respectively, to $v$ and $w$. Let $R_0$ be any connected region that is adjacent to both $v'$ and $w'$, let $D_0, D_1, \ldots, D_h$ denote the disks in the statement of \ref{geo_bdrycond} for $R_0$, and let $B_i = \partial D_i$. Suppose that $v'$ is in $B_i$. If $i = 0$, then there is a neighboring region $N$ that is {contained entirely} in $\intr(D_0)^c$ (by Lemma \ref{lem:enclosure}) and adjacent to $v'$. If $i > 0$, then there is a neighboring region $N$ that is {contained entirely} in $\intr(D_i)$ (by Lemma \ref{lem:enclosure}) and adjacent to $v'$. In either case, $w' \in B_i$ because $w'$ is also adjacent to $N$. Let $B_{i_1}, \ldots, B_{i_m}$ be the boundaries that contain $v'$. {{As we just showed,} it must also be true that} $w' \in B_{i_1}, \ldots, B_{i_m}$. If $m > 1$, then $w' \not\in B_{i_1} \cap \cdots \cap B_{i_m}$ because $D_{i_1} \cap \ldots \cap D_{i_m}$ is a single point by assumption \ref{geo_bdrycond}{; this} is a contradiction. This argument shows that if $v$ and $w$ have the same set of region adjacencies, then there is a unique $B_i$ that contains $v'$, there is a unique $B_j$ that contains $w$, and $B_i = B_j$.

Let $B$ be the disk boundary of $R$ that contains $v$ and $w$.} Either the interior of $R$ is contained in the region that is bounded by $B$ {or} it is contained in the complement of the region that is bounded by $B$. Without loss of generality, we suppose that the former is true. Let $\pi$ be the permutation of $\{1, \ldots, n\}$ such that the sequence of region adjacencies around $w$ is $\{R, N_{\pi(1)}, \ldots, N_{\pi(n)}\}$. Let $i_1, i_2 \in \{1, \ldots , n\}$, with $i_1 < i_2$, be a pair of indices. {By the argument above (with $R_0 = N_{i_1}$), there is a unique disk boundary $B_1$ for $N_{i_1}$ that contains $v'$ and $w'$. Similarly, there is a unique disk boundary $B_2$ for $N_{i_2}$ that contains $v'$ and $w'$. We have that $v', w' \in B_1 \cap B_2$.}

Because $B_1$ is homeomorphic to $S^1$, there exist paths $\gamma_1$ {and} $\gamma_2$ from $v'$ to $w'$ such that $\gamma_1 \cup \gamma_2 = B_1$. Because the interior of $N_{i_1}$ does not intersect $R$, it follows that $\gamma_1$ and $\gamma_2$ are both in the complement of the region that is bounded by $B'$. There are two paths from $v'$ to $w'$ on $B'$. Let $\tau$ be the unique choice of path such that $R$ is not contained in the region that is bounded by the closed curve $\tau \cup \gamma_1$. Either $\gamma_1$ is in the region that is bounded by the closed curve $\tau \cup \gamma_2$ or $\gamma_2$ is in the region that is bounded by the closed curve $\tau \cup \gamma_1$. Without loss of generality, we suppose that the latter is true.

Analogously to our argument above, there exist paths $\gamma_3$, $\gamma_4$ from $v'$ to $w'$ such that $\gamma_3 \cup \gamma_4 = B_2$ and $\gamma_3$ and $\gamma_4$ are in the complement of the region that is bounded by $B$. Because $B_2$ is homeomorphic to $S^1$, the paths $\gamma_3$ and $\gamma_4$ are either both contained in the region that is bounded by $\gamma_1 \cup \tau$ or both contained in the complement of the region that is bounded by $\gamma_2 \cup \tau$. Because $i_2 > i_1$, it must be the former case. Therefore, $\pi(i_2) < \pi(i_1)$. It follows that $\pi$ is order-reversing. If there were another vertex $x$ in $B$ that is adjacent to the same set of regions, then the orientation of those regions around $x$ would be the mirror of both the orientation of regions around $v$ and the orientation of regions around $w${. This} gives a contradiction when $n \geq 2$.
\end{proof}

For example, let $R$ be the region Koreatown in Figure~\ref{fig:geo_koreatown}. The two vertices that are shared by Koreatown and Little Bangladesh have the same region adjacencies, but they have mirrored orientations.

\section{Alternative Topological Approaches}\label{sec:alt_choices}

\subsection{0D Persistent Homology}\label{sec:0D}

{{We do} not compute 0D PH in the present {paper. However, it is appropriate to use 0D PH to study the structure of local extrema} when one is not interested in their geographical locations.}

Let $F$ be a real-valued function on a set $S$ of geographical regions. {In section \ref{sec:sublevel} (respectively, section \ref{sec:superlevel}), we described how one can analyze the local maxima (respectively, local minima) of $F$ by computing the 1D PH of the sublevel-set filtration (respectively, superlevel-set filtration). We now discuss how the 0D PH of the sublevel-set filtration (respectively, superlevel-set filtration) yields information about local minima (respectively, local maxima) of $F$.}

The 0D PH of the {sublevel-set} filtration encodes information about the structure of local minima of $F$ in a way that is similar to {how} 1D PH encodes information about the structure of local maxima. One can imagine taking $\alpha$-sublevel sets of the function in Figure \ref{fig:superlevel_ex} (where we display $\alpha$-\emph{super}level sets) to see why this is true. A region $R$ is a \emph{local minimum} if the value of $F(R)$ is less than the value of $F(N)$ for all neighboring regions $N$ of $R$ for which $N \cap R$ is 1D. If $R$ is a local minimum, there is a 0D homology class whose birth simplex is one of the vertices in one of the triangles in the preimage $g^{-1}(R)$. The class is born at filtration level $\alpha = F(R)$. For the LA data set of COVID-19 case rates, 0D homology classes correspond to regions that have a lower case rate than {neighboring} regions. The smaller the value of $F(R)$ in comparison to the {neighboring} regions, the more persistent the homology class is. There is also one infinite 0D homology class for each connected component. One can think of these classes as corresponding to a ``local minimum'' in the exterior region. However, unlike for 1D homology classes, there is no canonical map from 0D homology classes to regions because the birth simplex of a 0D class is a vertex that belongs to several regions. {Analogously, the} 0D PH of the {superlevel-set} filtration {encodes} information about the structure of local maxima of $F$. However, as with {a sublevel-set} filtration, there is no canonical map from 0D homology classes to regions. {Therefore, one cannot easily use the 0D PH of the sublevel-set filtration (respectively, superlevel-set filtration) to identify the geographical locations of the local minima (respectively, local maxima), so we did not examine 0D PH in our case studies.}


\subsection{Alternative Filtrations for Disconnected Geographical Spaces}\label{sec:exterior}

In section \ref{sec:sublevel} (respectively, section \ref{sec:superlevel}), we defined a {sublevel-set filtration (respectively, superlevel-set filtration)} in which we set the filtration values of all exterior-adjacent vertices and edges to the global minimum (respectively, {to the} {additive inverse of the} global maximum) of $F$. In applications in which the union of all regions is not connected, such as for the NYC zip codes in section \ref{sec:nyc}, an alternative definition is to consider extrema on each connected component separately, rather than on the entire geographical space at once. This solves the problem that an isolated region ({i.e.,} a geographical island\footnote{These are literal islands, rather than ``islands'' from a PH computation.}) is trivially both a local maximum and a local minimum because it is not adjacent to any other regions. In Definitions \ref{def:sublevel} and \ref{def:superlevel}, they appear as 1D homology classes that are born at the earliest filtration time{; this} may falsely emphasize the persistence of these trivial extrema.

\begin{definition}[Alternative {Sublevel-Set} Filtration]\label{def:altsublevel}
Let $\K$ be the \scom\ from section \ref{scom_construction} for a set $S$ of regions, and let $g$ be the assignment of 2D simplices to regions. {Additionally, let} $F: S \to \mathbb{R}$. If $\sigma$ is a vertex or edge on the boundary of $\K$, let $\tilde{\sigma}$ be the 2D simplex for which $\sigma$ is on the boundary of $\tilde{\sigma}$. On $\sigma$, we define the alternative {sublevel-set} filtration function $f$ to be
\begin{equation*}
    f(\sigma) = \min_R \{F(R) \mid R \subseteq C\}\,,
\end{equation*}
where $C$ is the connected component that contains {the region} $g(\tilde{\sigma})$. On all other simplices, the filtration function $f$ is equal to the {sublevel-set} filtration function.
\end{definition}

\begin{definition}[Alternative {Superlevel-Set} Filtration]\label{def:altsuperlevel}
Let $F: S \to \mathbb{R}$ for a set $S$ of regions. The alternative {superlevel-set} filtration function $f$ is the the alternative {sublevel-set} filtration function that is induced by $-F$.
\end{definition}

Definitions~\ref{def:altsublevel} and~\ref{def:altsuperlevel} are appropriate options if one seeks to treat each connected component independently. In these alternative definitions, each connected component uses only information about other regions in the same component. One then compares region values $F(R)$ to global extremum values on their connected components. One consequence of using these definitions is that one ignores isolated regions, which are trivial extrema. In Definitions~\ref{def:altsublevel} and~\ref{def:altsuperlevel}, these isolated extrema appear as points on the diagonal of a PD. This is often an appropriate way to handle isolated regions. However, when an isolated region is a global extremum of a data set, this may be undesirable. This situation never occurs in our data.

{NYC} has 14 connected components; several of {them} are zip codes that correspond to isolated islands. The alternative {sublevel-set and superlevel-set} filtrations effectively treat each connected component of NYC separately. In Figures~\ref{fig:nyc_sublevel_alt} and~\ref{fig:nyc_superlevel_alt}, we show the PDs that we compute using the alternative {sublevel-set and superlevel-set} filtrations that are induced by the vaccination-rate function that we defined in section \ref{sec:nyc}. In these PDs, we compare a zip code's per capita vaccination rate to the global minimum or maximum rate on its connected component, rather than {to} the global extremum in all of NYC. More precisely, the birth time of a connected component's global extremum is either the lowest per capita vaccination rate of that component (for the alternative {sublevel-set} filtration) or the additive inverse of the highest per capita vaccination rate of that component (for the alternative {superlevel-set} filtration). Consequently, the trivial island extrema yield homology classes on the diagonal {of a PD}.

\begin{figure}
    \centering
    \subfloat[Alternative sublevel filtration]{\includegraphics[width = .5\textwidth]{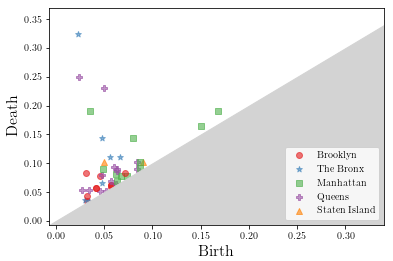}\label{fig:nyc_sublevel_alt}}
    \subfloat[Alternative superlevel filtration]{\includegraphics[width = .5\textwidth]{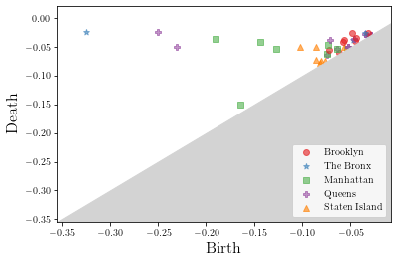}\label{fig:nyc_superlevel_alt}}
     \caption{PDs for the 1D PH of the NYC simplicial complex with filtrations that are induced by the per capita full vaccination rate by zip code on 23 February 2021. We show only the finite homology classes. {Each point in a PD corresponds to a non-isolated zip code, which we label according to its borough \cite{boroughs}, that has (a) a higher vaccination rate than its neighboring zip codes or (b) a lower vaccination rate than its neighboring zip codes.}
    }
\end{figure}

{The alternative sublevel-set filtration and the alternative superlevel-set filtration, along with their time-dependent versions, are implemented in our code at \url{https://bitbucket.org/ahickok/vineyard/src/main/}.}


\subsection{Multiparameter Persistent Homology}\label{sec:multi}

One can use {multiparameter persistent homology (MPH)} to study how the topology of a data set changes as one varies multiple parameters. {For a review of MPH, see \cite{multi}.} 

One can use {MPH} to study local extrema of functions that are nondecreasing {with} time. {To apply MPH} to our COVID-19 case-rate data, two feasible parameters are (1) time and (2) the cumulative COVID-19 case rate.  However, {MPH} {is difficult} to analyze{. Although} there are invariants (e.g., the rank invariant), {there is no complete discrete invariant} \cite{multi}. By contrast, one {can use PDs} for single-parameter PH.

\begin{definition}\label{def:multi}
Let $\K$ be the \scom\ from the construction in section \ref{scom_construction} for a set $S$ of regions. Let $F: \{t_0, \ldots, t_n\} \times S \to \mathbb{R}$ be a function {such that} $F(t, R) \geq F(s, R)$ for all $t \geq s$. Define the function $f(t_i, \sigma)$ to be the {sublevel-set} filtration that is induced by $F(t_i, \cdot)$. Let $\{\alpha_0, \ldots, \alpha_{\ell} \}$ be the image of $F$, where $\ell+1$ is the number of elements in the image. We define the bifiltration
\begin{equation*}
    \K_{i, j} := \begin{cases}
            \{ \sigma \in \K \mid f(t_i, \sigma) \leq \alpha_j \} \,, & i \in \{0, \ldots, n\}\,, \, j \in \{0, \ldots, \ell \} \\
            \K \,, & j > \ell \text{ and } i \geq 0 \\
            \K_{n, j} \,, & i > n \text{ and } j \geq 0 \\
            \emptyset \,, & i<0 \text{ or } j < 0 \,.
        \end{cases}
\end{equation*}
\end{definition}

One can use Definition~\ref{def:multi} to study cumulative COVID-19 case rates {as a function of} time.

\subsection{Multiparameter Zigzag Persistent Homology}\label{sec:zigzag}

{One can use \emph{multiparameter zigzag PH} {(MZPH)} to study how the topology of a data set changes as one varies multiple parameters nonmonotonically. See Section 2.1 of \cite{zigzag} for a short discussion of {MZPH}.}
 
{To use MZPH to study} our COVID-19 case-rate data, two feasible parameters are (1) time and (2) the current COVID-19 case rate. {A diagram of simplicial complexes, such as in Equation~\ref{eq:zigzag_diagram}, induces a diagram of homology groups.} This is a representation of a quiver. However, there are no known well-behaved statistical summaries (in contrast to single-parameter zigzag PH).

\begin{definition}\label{def:zigzag}
{Let $\K$ be the \scom\ from the construction in section \ref{scom_construction} for a set $S$ of regions}, and suppose that $F: \{t_0, \ldots, t_n\} \times S \to \mathbb{R}$. Define half steps $t_{i + 1/2} := t_i + (t_{i+1} - t_i)/2$ for $i \in \{0, \ldots, m-1\}$, and let $s_i := t_{i/2}$. Define the function $G: \{s_0, \ldots, s_{2n}\} \times S \to \mathbb{R}$ as follows:
\begin{align*}
    G(s_i, R) &= \begin{cases}
        F(s_i, R) \,, & i \text{ is even} \\
        \max\{ F(t_{(i-1)/2}, R), F(t_{(i+1)/2}, R) \} \,, & i \text{ is odd}\,.
    \end{cases}
\end{align*}

We define the function ${h}(s_i, \cdot)$ to be the {sublevel-set} filtration that is induced by $G(s_i, \cdot)$. Let $\{\alpha_0, \ldots, \alpha_{\ell} \}$ be the image of $G$. We define
\begin{equation*}
    \K_{i, j} := \begin{cases}
        \{ \sigma \in \K \mid {h}(s_i, \sigma) \leq \alpha_j \} \,, & i \in \{0, \ldots, 2n\}\,, \, j \in \{0, \ldots, \ell \} \\
        \K \,, & j > \ell \text{ and } i \geq 0\\
        \K_{2n, j} \,, & i > 2n \text{ and } j \geq 0 \\
        \emptyset \,, & i<0 \text{ or } j < 0\,.
    \end{cases}
\end{equation*}
This yields the following diagram:
\begin{equation}\label{eq:zigzag_diagram}
\begin{tikzcd}
{}                                                                      & {}                                                                                      & {}                                                                      & {}                                                                      &    \\
{K_{\alpha_0, s_3}} \arrow[d, hook] \arrow[u, hook] \arrow[r, hook]      & {K_{\alpha_1, s_3}} \arrow[d, hook] \arrow[u, hook] \arrow[r, hook]                        & {K_{\alpha_2, s_3}} \arrow[d, hook] \arrow[u, hook] \arrow[r, hook]       & {K_{\alpha_3, s_3}} \arrow[d, hook] \arrow[u, hook]                 &    \\
{K_{\alpha_0, s_2}} \arrow[r, hook]                                     & {K_{\alpha_1, s_2}} \arrow[r, hook]                                                     & {K_{\alpha_2, s_2}} \arrow[r, hook]                                     & {K_{\alpha_3, s_2}} \arrow[r, hook]                                     & {} \\
{K_{\alpha_0, s_1}} \arrow[d, hook] \arrow[u, hook] \arrow[r, hook] & {K_{\alpha_1, s_1}} \arrow[d, hook] \arrow[u, hook] \arrow[r, hook] \arrow[d, hook] & {K_{\alpha_2, s_1}} \arrow[d, hook] \arrow[r, hook] \arrow[u, hook] & {K_{\alpha_3, s_1}} \arrow[d, hook] \arrow[u, hook] \arrow[r, hook] & {} \\
{K_{\alpha_0, s_0}} \arrow[r, hook]                                     & {K_{\alpha_1, s_0}} \arrow[r, hook]                                                     & {K_{\alpha_2, s_0}} \arrow[r, hook]                                     & {K_{\alpha_3, s_0}} \arrow[r, hook]                                     & {}
\end{tikzcd}\,.
\end{equation}
The inclusion maps induce a corresponding diagram of homology groups.
\end{definition}

One can use Definition~\ref{def:zigzag} to study non-cumulative COVID-19 case rates {as a function of} time.


\section{The Full LA Vineyard}\label{sec:fullLA}

In Figure \ref{fig:la_vy_full}, we {show} the full vineyard {that we} discussed in section \ref{sec:LA}.

\begin{figure}
    \centering
    \subfloat[]{\includegraphics[height =.34\paperheight]{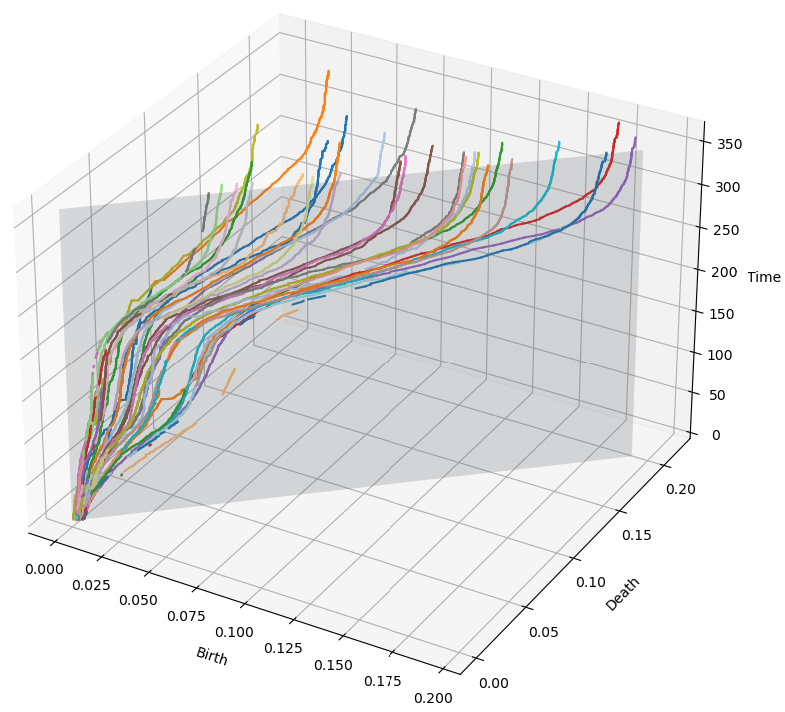}} \\
    \subfloat[]{\includegraphics[height = .34\paperheight]{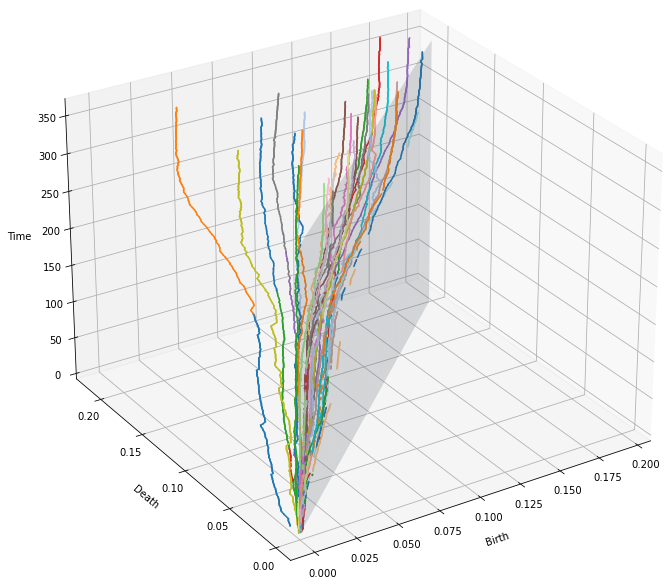}}
    \caption{(a) The vineyard for the LA simplicial complex that we construct using the sublevel-set filtration from the 14-day mean per capita case rate during the period 25 April 2020--25 April 2021. Each vine {is associated with} a COVID-19 anomaly. We color each vine according to the geographical location(s) of its associated anomaly. (See Figure \ref{fig:la_vy_legend} for the legend.) Because the geographical location of an anomaly can change {with} time, a single vine can have multiple colors. (b) A different view of the same vineyard.}
    \label{fig:la_vy_full}
\end{figure}

\begin{figure}
    \centering
    \includegraphics[width = \textwidth]{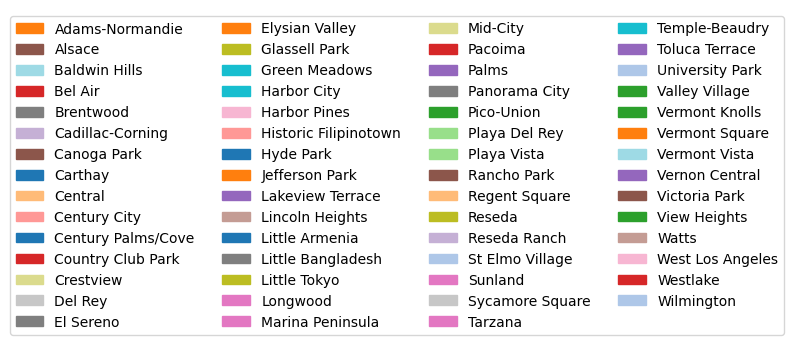}
    \caption{The legend for Figure \ref{fig:la_vy_full}. Each of the depicted regions is a local maximum of the COVID-19 case-rate function for some subset of the time period 25 April 2020--25 April 2021.}
    \label{fig:la_vy_legend}
\end{figure}


{\section{{Results of an} All-But-One Statistical Test}

Previously, we {examined} local extrema of real-valued geospatial data; we called these ``anomalies''. For real-valued geospatiotemporal data, one can alternatively {examine} a different notion of anomaly. {In this context}, we say that a region is an \emph{anomaly} if one is not able to infer its data successfully from the data of the other regions. More precisely, let $X$ be the matrix whose $(i, j)$th entry is the value of region $j$ at time step $i$. In our case study of COVID-19 case rates {in LA}, the regions are the neighborhoods of LA and the $(i, j)$th entry of $X$ is the 14-day mean per capita case rate in region $j$ on the $i$th day after $25$ April 2020. {Let $\bm{x^j}$ denote the $j$th column of $X$, and let $X^j$ denote the matrix that one obtains by deleting column $\bm{x^j}$. The vector $\bm{x^j}$ has the data for region $j$, and the matrix $X^j$ has the data for all regions except for region $j$.} We define our \emph{prediction} of region $j$ to be the least-squares solution {$\bm{b^*}$} to $X^j \bm{b} = \bm{x^j}$, and we quantify the predictability of region $j$ by calculating the relative residual norm $\norm{X^j\bm{b^*} - \bm{x^j}}_2\left.\right/\norm{\bm{x^j}}_2$. A smaller relative residual norm indicates greater predictability.

In Figure \ref{fig:allbutone}, we show the result of this ``all-but-one'' statistical test for the LA COVID-19 {data set}. In this figure, we plot the relative residual norm for each neighborhood. All neighborhoods have {near-$0$} relative residual norms, {so the} neighborhoods' case rates are {very} predictable {when one knows the case rates {of} all other} neighborhoods. The mean relative residual norm is only $5.970\times 10^{-7}$, with a standard deviation of $\sigma {\approx} 7.558\times 10^{-7}$. The neighborhoods {with the least predictability} ({specifically}, those whose relative residual norms have a z-score {that is} {larger} than 3) are Brookside, Little Armenia, Little Tokyo, Sycamore Square, and Toluca Terrace. We {show} their relative residual norms and z-scores in Table \ref{tab:allbutone}.

\begin{figure}
    \centering
    \includegraphics[scale=.8]{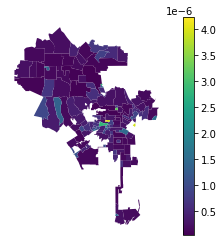}
    \caption{The results of {an} all-but-one statistical test for the LA COVID-19 case-rate data. We plot the relative residual norm for each neighborhood.}
    \label{fig:allbutone}
\end{figure}

\begin{table}[]
    \centering
    \begin{tabular}{||c|c|c||}
        \hline
         Neighborhood &  Relative Residual Norm & z-score \\
         \hline \hline
         Brookside & 3.973$\times 10^{-6}$ & 4.466	  \\
         \hline
         Little Armenia & 3.220$\times 10^{-6}$ & 3.471 \\
         \hline
         Little Tokyo & 3.944$\times 10^{-6}$ & 4.429 \\
         \hline
         Sycamore Square & 3.944$\times 10^{-6}$ & 4.429 \\
         \hline
         Toluca Terrace & 2.873$\times 10^{-6}$ & 3.012\\
         \hline
    \end{tabular}
    \caption{The relative residual norms and z-scores for the LA neighborhoods that are least predictability according to our all-but-one test.}
    \label{tab:allbutone}
\end{table}

{The difference between what we learn from the all-but-one statistical test and what we learn from our TDA approach is the following.} {Using} our TDA approach, we identified local extrema (i.e., regions whose values are {either all larger than or all smaller} than {those of} all neighboring regions); this is a geographical notion of anomaly. {By contrast,} the all-but-one statistical test {does not inherently} capture local extrema because the test does not consider geographical adjacencies. Despite this conceptual difference, {we observe} some overlap between the anomalies that the two approaches identify. For example, the neighborhoods Little Tokyo and Little Armenia {are} identified as anomalies by both approaches. For further {examples, compare} Figure \ref{fig:allbutone} with Figure \ref{fig:la_vy_map}.


\section{Demographic Data}\label{sec:demographic}

We provide some demographic data for NYC and LA for {readers} who are interested in studying patterns between the {PDs} and demographic data, {although an investigation of such patterns} is beyond the scope of the present paper. In Figure \ref{fig:median_income}, we plot the median income for each zip code\footnote{{We do not possess} median income data for LA zip codes 90073, 90089, 90095, 91330, 91522, {and} 91608. These zip codes are {in} non-residential areas. {For example,} 90073 {corresponds to} the Veterans Administration.} \cite{census_5year}. The geographical boundaries of the NYC and LA zip codes are given by the {\sc shapefile}s \cite{NYCshp} and \cite{LAzipshp}, respectively. It is worthwhile to examine and compare other demographic data (such as racial, religious, and political data) to the PDs.

\begin{figure}
    \centering
    \subfloat[]{\includegraphics[width=.5\textwidth]{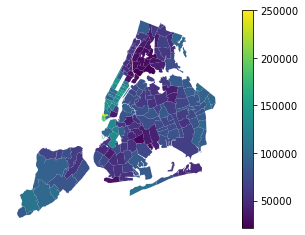}}
    \hspace{5mm}
    \subfloat[]{\includegraphics[width = .4\textwidth]{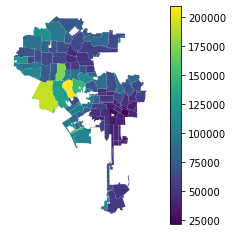}}
    \caption{{(a) Median household income by zip code in NYC. (b) Median household income by zip code in LA.}}
    \label{fig:median_income}
\end{figure}
}


\section*{Acknowledgements}

We thank Henry Adams, Heather Zinn Brooks, Michelle Feng, Lara Kassab, and Nina Otter for helpful discussions. Additionally, we are grateful to Michelle Feng for teaching us how to work with geospatial data. We thank the Los Angeles County Department of Public Health for providing the LA city data on COVID-19 and the population estimates of LA neighborhoods.



\end{document}